%% file: FB_Vola_final_d.tex
\documentclass[12pt,a4paper,english,openright]{article}

\usepackage{graphicx}
\usepackage{pgfplots}
\usepackage{color} 
\definecolor{grey}{gray}{0.90}
\usepackage{a4wide}
\usepackage{verbatim}
\usepackage{amssymb,amsthm,amsmath,amsfonts} 
\usepackage[mathscr]{eucal} 
\usepackage{colonequals, enumerate}
\usepackage{todonotes}
\usepackage{bbm}
\usepackage{authblk}
\newcommand*{\N}{\mathbb{N}} 
\newcommand*{\R}{\mathbb{R}} 

\renewcommand*{\d}{\mathrm{d}} 
\newcommand*{\e}{\mathrm{e}} 
\renewcommand*{\epsilon}{\varepsilon} 

\newcommand*{\abs}[1]{\left|{#1}\right|} 
\newcommand*{\pnorm}[1]{\left\|{#1}\right\|} 
\newcommand*{\pinner}[2]{\langle{#1},{#2}\rangle} 
 
\newcommand{\id}{1\hspace{-0,9ex}1} 
\newcommand{\hy}{\ensuremath{\hat{Y}}}
\newcommand{\tw}{\ensuremath{\tilde{W}}} 
\newcommand{\ts}{\ensuremath{\tilde{\sigma}}} 

\DeclareMathOperator{\E}{\mathbb{E}} 

\DeclareMathOperator{\diag}{diag}
\renewcommand*{\P}{\mathbb{P}} 
\newcommand*{\condE}[2]{{\E  \left[{#1 \left|\right. #2}\right]}} 
\newcommand*{\norm}[1]{\|{#1}\|} 

\newcommand{\msm}{\mathrm{MSM}}
\newcommand{\hmm}{\mathrm{HMM}}
\newcommand{\sv}{\mathrm{SV}}

\newcommand{\F}{\mathcal{F}}

\newcommand{\G}{\mathcal{G}}

\newtheoremstyle{noindent} 
                        {12pt}    
                        {12pt}    
                        {\itshape}         
                        {}         
                        {\scshape}
                        {.}        
                        { } 
                        {}         

\newtheoremstyle{noindent2} 
                        {12pt}    
                        {12pt}    
                        {\normalfont}         
                        {}         
                        {\scshape}
                        {.}        
                        { } 
                        {}         

\theoremstyle{noindent}

\newtheorem{theorem}{Theorem}[section]
\newtheorem{lemma}[theorem]{Lemma}

\newtheorem{proposition}[theorem]{Proposition}

\theoremstyle{noindent2}
\newtheorem{definition}[theorem]{Definition}
\newtheorem{remark}[theorem]{Remark}

\definecolor{myblue}{RGB}{30,144,255}

\begin{document}
\title{Filterbased Stochastic Volatility in\\ Continuous-Time Hidden Markov Models}

\author[1]{Vikram Krishnamurthy}
\author[2]{Elisabeth Leoff}
\author[3]{J\"orn Sass}
\affil[1]{Department of Electrical and Computer Engineering, University of British Columbia, Vancouver, Canada, vikramk@ece.ubc.ca}
\affil[2]{Department of Financial Mathematics, Fraunhofer Institute for Industrial Mathematics ITWM, Kaiserslautern, Germany, elisabeth.leoff@itwm.fraunhofer.de}
\affil[3]{Department of Mathematics, University of Kaiserslautern, Kaiserslautern, Germany, sass@mathematik.uni-kl.de}
\date{February 5, 2016}

\maketitle

\begin{abstract}
Regime-switching models, in particular Hidden Markov Models (HMMs) where the switching is driven by an unobservable Markov chain, are widely-used in financial applications, due to their tractability and good econometric properties. In this work we consider HMMs in continuous time with both constant and switching volatility. In the continuous-time model with switching volatility the underlying Markov chain could be observed due to this stochastic volatility, and no estimation (filtering) of it is needed (in theory), while in the discretized model or the model with constant volatility one has to filter for the underlying Markov chain. The motivations for continuous-time models are explicit computations in finance. To have a realistic model with unobservable Markov chain in continuous time and good econometric properties we introduce a regime-switching model where the volatility depends on the filter for the underlying chain and state the filtering equations. We prove an approximation result for a fixed information filtration and further motivate the model by considering social learning arguments. We analyze its relation to the switching volatility model and present a convergence result for the discretized model. We then illustrate its econometric properties by considering numerical simulations.
\end{abstract}

\section{Introduction}

Regime-switching models are very popular in the field of mathematical finance to describe return processes with time-changing drift or volatility parameters. They are a possible way to generalize the classical Black-Scholes log normal stock price model by making the parameters dependent on a Markov chain with finitely many states. This Markov chain can be interpreted as encoding all the information in the market into one single process that describes the underlying state of the economy. As such this chain is assumed to be unobservable, the normal investor does not have access to all the possible information. The only observation that can be made is the return process itself.\\
In this work we consider both the Hidden Markov Model (HMM) and what we call the Markov Switching Model (MSM) in continuous time. Both models have been studied extensively and are popular for applications. The difference between them is only that the HMM has constant volatility, while in the MSM the volatility jumps with the Markov chain. This difference leads to a difference in the behaviour of information between both models, a distinction which is not present in discrete time, as pointed out in \cite{Elliott2008}. For comparison of estimation procedures for the continuous-time MSM we refer to \cite{Hahn2010}.\\
In finance, one often needs models that are defined in continuous time since these allow for explicit results. To be applicable these have to be close to the results for the real-world discrete-time setting. 
For example, in portfolio optimization the aim is to find the strategy that maximizes the expected utility at terminal time. These optimal strategies can be derived explicitly only in continuous time, but in the HMM they are close to the discrete-time strategies, cf.\ the detailed discussion in Section \ref{sec_po}. Consequently, the continuous-time HMM is such a model which is very well tractable and yields results close to the discrete-time model. Thus we may call it discretization-consistent.  \\
The volatility in the continuous-time HMM is assumed to be constant. Empirical observations of financial data across a wide range of instruments, markets and time periods indicate that for many real data sets this is not true, as was already noted by e.g. Fama in 1965 (\cite{Fama1965}). This leads to the concept of stochastic volatility in model building. There have been many different modeling approaches leading to various models where the volatility itself is a stochastic process. Popular models are e.g. Engle's famous Autoregressive Conditional Heteroskedasticity model (ARCH) from 1982, see Engle (\cite{Engle1982}), and its generalization GARCH, see Bollerslev (\cite{Bollerslev1986}), or in continuous time the well-known Heston model (\cite{Heston1993}).\\
In the context of regime-switching models the natural way to introduce a non-constant volatility is to use the Markov chain to control the volatility, as usually done in discrete time when considering regime switching models, see Hamilton \cite{Hamilton1989}. Switching can also be introduced in the  volatility models mentioned above.  Switching between constant values for the volatility  leads in continuous time to the MSM, which has better econometric properties than the HMM for many applications. Some stylized facts such as volatility clustering cannot be reproduced in the HMM. However, it turns out that in the continuous-time MSM the Markov chain can be observed using the quadratic variation of the return process, see Proposition \ref{prop_no_filter_MSM} below. Therefore, results obtained in the MSM can be far away from results obtained in the corresponding discrete-time model. Thus we do not have the discretization-consistency for the continuous-time MSM as we have it for the HMM.\\
To get a model which is both consistent and has good econometric properties, the idea now is to introduce stochastic volatility in a continuous-time HMM. Haussmann and Sass \cite{Sass2004/2} propose a regime-switching model where the volatility depends on an observable diffusion. They prove that despite the non-constant volatility in this model the chain is still hidden and derive its filtering equation which is again finite-dimensional. The question of how to actually model the volatility process is still left open.
Due to the better econometric properties of the MSM our goal is to find a volatility process that leads to an HMM that in a sense approximates the continuous-time MSM.\\
Thus in this work we introduce an HMM in continuous time where the volatility is a linear function of the (normalized) filter. For a given filtration we construct the volatility process adapted to this filtration that results in the best approximation of the MSM-returns in the mean squared sense. This result motivates the introduction of the Filterbased-Volatility Hidden Markov Model (FB-HMM), where the volatility is a linear function of the normalized filter. In this model the returns and the filter constitute a system of equations coupled by the volatility process. In a sense, this model lies between the HMM and MSM: The chain is still hidden as in the HMM, but the volatility is non-constant similar to the MSM. The function that connects the filter to the volatility process in the FB-HMM is even the same function that connects the chain to the volatility in the MSM. Another motivation for the model comes from an instantaneous iteration of observing the returns and trading (resulting in a new volatility) in a financial market, which is close to social learning arguments, see Remark \ref{rem_social_learning}. \\
But as a function of the unnormalized filter the volatility in the FB-HMM does not satisfy the assumptions in \cite{Sass2004/2}. In this work we state the filtering equations for the FB-HMM, following the methodology from \cite{Elliott1993}, and see that the filter is as tractable as in the HMM from \cite{Sass2004/2}.
Further, the issue of how the model behaves under discretization is very important for practical applications, as we have already mentioned. For the question of approximation and consistency in the FB-HMM we consider the discretized returns. Here, the issue of consistency is a different one than what is typically investigated in statistics. The discretized FB-HMM depends on the continuous-time filter at the discretization points through its volatility, so it is not truly a discrete-time model. For consistency we consider the convergence of the discretized returns instead and prove that already for the simple Euler discretization the global discretization error converges to 0 in $L^2$.\\
We proceed as follows: In Section \ref{sec_def_HMM_MSM} we introduce the HMM and MSM in continuous time, prove the observability of the Markov chain in the continuous-time MSM and consider portfolio optimization in regime-switching models. In Section \ref{sec_fb_vola} we introduce the FB-HMM and prove an approximation result for a fixed information filtration. We further motivate the model by considering social learning arguments and shortly investigate the stylized facts present in the returns. In Section \ref{sec_consistency} we then prove that the Euler-discretization of the FB-HMM returns converges to the continuous-time returns. We conclude the paper by presenting some numerical results.

\section{The HMM and MSM in Continuous Time} \label{sec_def_HMM_MSM}
Consider a return process $(R_t)_{t\in [0,T]}$ with dynamics
\[
\d R_t= \mu_t\d t+\sigma_0\d W_t,
\]
where the unobservable drift process $\mu_t=\mu(Y_t)$ jumps between $d$ levels following a process $Y$. The volatility $\sigma_0\in\R_{>0}$ is constant and $(W_t)_{t\in[0,T]}$ is a standard Brownian motion. Since $\mu_t$ has only finitely many states, we can represent it as $\mu_t=\mu^TY_t$ with $\mu\in\R^{d}$ the so-called state vector. $Y$ is assumed to be an irreducible, time-homogenous Markov chain on $\{e_1,\dots,e_d\}$ with rate matrix $Q\in\R^{d\times d}$. We assume that $Y_0\sim\nu $ with $\nu$ the invariant distribution of $Y$ that exists and is unique due to the irreducibility and finite state space of $Y$.

This model seems natural if we think of the drift as influenced by discrete events such as the arrival of news or shocks that change the overall state of the economy. The Markov chain then models this underlying state of the economy that influences the behaviour of stock returns.

We assume that the Markov chain $Y$ cannot be observed, only the returns can be seen by the investor. Thus we are in the setting of incomplete information. This is why the model is called Hidden Markov Model (HMM). Any investment decision can be made only with respect to the information available from observing $R$, i.e.\ must be adapted to the filtration $\F^R_t=\sigma(R_s\vert s\leq t)$. In the terminology of filtering theory, $R$ is the observation while $Y$ is the signal and we are interested in the conditional expectation $\hat{Y}_t\colonequals\condE{Y_t}{\F^R_t}$.

The filtering equation for the state process is given by the Wonham filter, see \cite{Wonham1965}. It is a finite dimensional equation describing the dynamics of $\hat{Y}_t$:
\[
\d\hat{Y}_t=Q^T\hat{Y}_t\d t+\diag(g)\hat{Y}_t-(g^T\hat{Y}_t)\hat{Y}_t\d V_t,
\]
with $V_t=\int_0^t\sigma_0^{-1}\mu^T(Y_s-\hat{Y}_s)\d s+W_t$ the innovation process, an $\F^R$-adapted $\P$-Brownian motion, and $g^T=\sigma_0^{-1}\mu^T$.

Since this equation is not linear in the filter $\hat{Y}$ one typically uses the Zakai equation for the so-called unnormalized filter which is the conditional expectation under a new measure. The Zakai equations for the unnormalized filters needed for the EM-algorithm were introduced by Elliott in \cite{Elliott1993}.

For these filters we need the martingale density process 
\[
Z_t=\exp\left(-\int_0^t\sigma_0^{-1}\mu^TY_s\d W_s-\frac{1}{2}\int_0^t\pnorm{\sigma_0^{-1}\mu^TY_s}^2\d s\right)
\]
and with this the new measure $\tilde{\P}$ via $\frac{\d \tilde{\P}}{\d\P}=Z_T$. Due to the boundedness of $\sigma_0^{-1}\mu^TY_t$ Girsanov's Theorem guarantees that $\tilde{W}$ defined as
\[
\d\tilde{W}=\sigma_0^{-1}\mu^TY_t\d t+\d W_t
\]
is a $\tilde{\P}$-Brownian motion. Using the new measure we can define the unnormalized filter $\rho$ as
\[
\rho_t=\tilde{\E}[Z_T^{-1}Y_t\vert\F_t^R]
\]
from which the normalized filter can be derived using the Kallianpur-Striebel formula, which is an abstract version of Bayes' formula:
\[
\hat{Y}_t=\frac{\rho_t}{\tilde{\E}[Z_T\vert\F_t^R]}.
\]
Observing that
\[
1=\condE{1}{\F_t^R}=\condE{1^TY_t}{\F_t^R}=1^T\condE{Y_t}{\F_t^R}=1^T\hat{Y}_t=\frac{1^T\rho_t}{\tilde{\E}[Z_T\vert\F_t^R]}
\]
we see that
\[
\hat{Y}_t=\frac{\rho_t}{1^T\rho_t}.
\]
Thus for filtering purposes it is enough to calculate the unnormalized filter $\rho$ which has the advantage that it follows a linear equation driven by the observation process (see \cite{Elliott1993})
\[
\d\rho_t=Q^T\rho_t\d t+\frac{1}{\sigma_0^2}\diag(\mu)\rho_t\d R_t.
\]

One drawback of the HMM is that it assumes the returns to move with a constant volatility. Since we model the drift as depending on the Markov chain describing the current state of the economy it is only natural to model the volatility as depending on the same Markov chain. This leads us to what we call the Markov Switching Model (MSM), as one does usually in the discrete-time model.

\subsection{The Markov Switching Model} \label{sec_MSM}

In the MSM we assume that the volatility has different states that depend on the Markov chain $Y$ in the same way the drift does, i.e.\ 
\[
\d R_t= \mu_t\d t+\sigma_t\d W_t,
\]
with
\begin{align*}
\mu_t&=\mu^TY_t,\quad \mu\in\R^{d}\\
\sigma_t&=\sigma^TY_t, \quad \sigma\in\R^{d}_{>0}.
\end{align*}

The main feature of the MSM is that it allows for a changing volatility process over time. For the case of $\sigma_i=\sigma_j$ for all $i,j$ it can be seen as a generalization of the HMM. One property of the model is that in continuous time the dependence of the volatility on the same Markov chain as the drift makes the Markov chain observable, as we prove in the following proposition.

\begin{proposition}\label{prop_no_filter_MSM}
Let $\d R_t=\mu^TY_t\d t+\sigma^TY_t\d W_t$. Then the signal $Y$ is adapted to $\F^R$. In particular $\E[Y_t\vert\F_t^R]=Y_t$, i.e.\ there is no filtering problem.
\end{proposition}
\begin{proof}
Fix $t\in [0,T]$ and let $n\in \N$. Due to the exponentially distributed waiting times of $Y$ for almost all $\omega\in\Omega$ we can find some $\epsilon_{\omega}>0$ such that
\[
Y_t(\omega)=Y_{t+\epsilon}(\omega)
\]
for $0\leq\epsilon\leq\epsilon_{\omega}$. Thus the following limit equality holds almost surely
\begin{align*}
(\sigma^TY_t)^2&=\lim_{m\to\infty}\frac{1}{\epsilon_m}\int_t^{t+\epsilon_m}(\sigma^TY_t)^2\d s=\lim_{m\to\infty}\frac{1}{\epsilon_m}\int_t^{t+\epsilon_m}(\sigma^TY_s)^2\d s\\
&=\lim_{m\to\infty}\frac{1}{\epsilon_m}([R]_{t+\epsilon_m}-[R]_t),
\end{align*}
where $\epsilon_m$ is some sequence converging to $0$ from above and $[R]_t=\int_0^t(\sigma^TY_s)^2\d s$ is the quadratic variation of $R$. There exist $M_n$ such that $\epsilon_m<\frac{1}{n}$ for all $m>M_n$, so $\frac{1}{\epsilon_m}([R]_{t+\epsilon_m}-[R]_t)$ is measurable with respect to $\F^R_{t+\frac{1}{n}}$ for $m>M_n$. In particular it follows that the limit $(\sigma^TY_t)^2$ is $\F^R_{t+\frac{1}{n}}$-measurable.
Since $\sigma_i\neq\sigma_j$ for $i\neq j$ and $\sigma_i>0$ this implies that $Y_t$ is measurable with respect to $\F^R_{t+\frac{1}{n}}$ for all $n\in\N$, so also with respect to the intersection $\cap_{n\in\N}\F_{t+\frac{1}{n}}^R$. But $\F^R$ is a right-continuous filtration, so
\[
\bigcap_{n\in\N}\F_{t+\frac{1}{n}}^R=\bigcap_{\epsilon>0}\F_{t+\epsilon}^R=\F_t^R
\]
and $Y$ is indeed adapted to $\F^R$.
\end{proof}

This result has a strong implication for trying to fit the continuous-time MSM to real data: We can observe (discretized) returns on the stock market, but there is no process available that directly corresponds to the Markov chain.\\
The chain describes the general ``state of the economy'', a quantity that is influenced by a lot of different factors with different characteristics, like high-level political decisions, central banks changing their interest rates, the development of real estate prices and similar indices. It is not feasible to quantify all these factors and distill their information into one process $Y$, so for our model we just assumed the chain to be unobservable. But Proposition \ref{prop_no_filter_MSM} tells us that this is not possible for a regime-switching model in continuous time with switching volatilities, the theory always implies that the chain is adapted to the observation filtration. This difference between the HMM and MSM is not present in discrete time. In the discrete-time models, $Y$ is not observable for both constant and switching volatility. Thus, even with careful discretization the information from observing the discrete-time MSM might be far away from the continuous-time model. Consequently, the results obtained in the continuous-time model might not be close the results in the discrete-time MSM. In this sense, the MSM is not consistent.\\

\subsection{Portfolio Optimization in HMM and MSM}\label{sec_po}

To illustrate the difference of HMM and MSM as motivated in the introduction we consider
as application in finance  a portfolio optimization problem. We look at $R$ above as stock returns and introduce further a money market account with constant interest rate $r=0$. We fix a time horizon $T>0$ and  
describe trading by the fraction $\pi_t$ of the portfolio value (wealth) $X^\pi_t$ invested at $t\in[0,T)$ in stocks. For initial capital $x_0$, the wealth controlled by $\pi =(\pi_t)_{t\in[0,T)}$ evolves as
\[
\d X^\pi_t = X^\pi_t(1-\pi_t) r\, \d t + X^\pi_t \pi_t \d R_t = X^\pi_t \pi_t \d R_t, \quad X^\pi_0 = x_0.
\]
An investor whose preferences are given by a utility function $U$, likes to maximize expected utility of terminal wealth, 
\[
\E[U(X^\pi_T)],
\]
over all {\em admissible} strategies $\pi$. Typical utility functions are 
\[
  U_\alpha(x) = \frac{x^\alpha}{\alpha}, \, \alpha < 1, \alpha\not=1 \quad 
\text{ or } \quad U_0(x) = \log(x).
\]
%
%
%

Let $\hat Z_t = \E[Z_t\,|\,\F_t^R]$ denote the conditional density. 
Then Sass and Haussmann \cite{Sass2004} show that in the HMM for utility function $U_\alpha$
\[ \pi^\hmm_t \hspace{-2pt} =\hspace{-2pt} \frac{1}{(1-\alpha) \E\left[ 
 {\hat Z}^{\frac{\alpha}{\alpha-1}}_{t,T}\,|\,\rho_t\right]} 
 \Bigg\{
 \frac{1}{\sigma_0^2} \mu^T\hat{Y}_t \E\hspace{-2pt}\left[{\hat{Z}}^{\frac{2\alpha-1}{\alpha-1}}_{t,T}\,|\,\rho_t\right]
 + \frac{1}{\sigma_0} \E\hspace{-2pt}\Big[{\hat Z}^{\frac{2\alpha-1}{\alpha-1}}_{t,T}\hspace{-2pt} \int_t^T \hspace{-4pt} 
 (D_t\rho_{t,s}) \mu\, \frac{1}{\sigma_0^2} \d R_s \Bigm|\rho_t \Big]\Bigg\}
\]
is optimal, where  ${\hat Z}_{t,T} = {\hat Z}_T/{\hat Z}_t$ and
$ \rho_{t,s} = \rho_s/ \hat{Z}^{-1}_t$. Further $D_t\rho_{t,s}$ is the Malliavin derivative which can be shown  to satisfy a linear SDE due to the linearity of the Zakai equation for $\rho_t$. 
Note that for $U_0=\log$ we get
\[
\pi^\hmm_t = 
\sigma_0^{-2} {\mu}^T{\hat Y_t}.
\]
In the MSM B\"auerle and Rieder \cite{Baeuerle2004} show that for observable $Y$ (which is the case for an MSM due to Proposition \ref{prop_no_filter_MSM}) 
\[
\pi_t^\msm = \frac1{1-\alpha}\, \frac{\mu^T {Y_t}}{({\sigma}^T{Y_t})^2}
\]
is optimal for $U_\alpha$ with $\alpha=0$ corresponding to $U_0=\log$.

%
%
%
%
%
%
Without providing the filtering equations and the model formulation of the more classical discrete-time models here, let us denote the filters in the discretized HMM with constant volatility and in the discretized MSM with switching volatility  by  
$\hat Y_k^\hmm$ and $\hat Y_k^\msm$, respectively. We obtain for mild parameters (fraction staying in [0,1]) that of order $\Delta t$ the optimal strategies for maximizing logarithmic utility are given by risky fractions 
\[
\pi_k^\hmm= \frac{{\mu}^T{\hat Y_k^\hmm}}{\sigma_0^2}+o(\Delta_t), \quad 
\pi_k^\msm= \frac{{\mu}^T{\hat Y_k^\msm}}{({\sigma}^T{\hat Y^\msm_k})^2}+o(\Delta_t),
\]
where $\Delta_t$ denotes the discretization step, see Taksar and Zeng \cite{Taksar2007}.
Note that in $\pi_k^\msm$ the filter $\hat Y^\msm_k$ appears while the corresponding continuous-time strategy $\pi^\msm_t$ is based on the Markov chain $Y_t$ itself. This leads to a large deviation of the corresponding optimization results which is not desirable. For the HMM both are based on the filter, the discrete-time filter being the discretization of the continuous-time filter, see \cite{James1996}. The optimization results for the continuous-time HMM are thus close to the discrete time model, if we look at strategies constrained to $[0,1]$ meaning that neither short selling nor borrowing are allowed. In this sense the HMM is discretization-consistent. 

\bigskip

This predicament of the MSM, the lack of discretization-consistency, is the main motivation for introducing an HMM with observable volatility as we will in the next section.

\section{A Filterbased Volatility Model}\label{sec_fb_vola}
We will now consider a continuous time HMM for stock returns that was introduced by Haussmann and Sass in \cite{Sass2004/2}. As in the classical HMM the return process $R$ depends on the Markov chain $Y$ with state space $\{ e_1,\dots,e_d\}$ via
\[
R_t=\int_0^t\mu^TY_s\d s+\int_0^t\sigma_s\d W_s,
\]
with $W$ a Brownian motion. The drift process of the returns is again given by $\mu^TY_t$, with the state vector $\mu\in\R^{d}$ containing its $d$ possible states.

The difference to the classical HMM is that the volatility process $\sigma_t$ is not constant anymore. We assume that $\sigma_t$ is uniformly bounded with uniformly bounded inverse $\sigma_t^{-1}$.

Since $\sigma_t^{-1}$ is uniformly bounded the process $\theta_t=\sigma_t^{-1}\mu^TY_t$, which is the market price of risk, is uniformly bounded as well and the density process $\d Z_t =-Z_t\theta_t\d W_t$ is a martingale. Thus we can define the reference measure $\tilde{\P}$ for filtering in this model as $\d \tilde{\P}=Z_T\d\P$. Similar to the HMM with constant volatility Girsanov's Theorem guarantees that $\d\tilde{W}=\theta_t\d t+\d W_t$ is a $\tilde{\P}$-Brownian motion.

For the volatility Haussmann and Sass \cite{Sass2004/2} specify a Markovian model with an observable diffusion $\xi$ driven by $\tilde{W}$:
\[
\d\xi_t=\nu(\xi_t)\d t+\tau(\xi_t)\d\tilde{W}_t,
\]
with $\nu,\tau$ continuously differentiable functions of the appropriate dimensions with bounded partial derivatives. The volatility process is then given by
\[
\sigma_t=\bar{\sigma}(\xi_t)
\]
where the function $\bar{\sigma}$ has bounded partial derivatives.

This choice of volatility model preserves the non-observability of the Markov chain from the HMM while having non-constant volatility as in the MSM. In \cite{Sass2004/2} the filtering equation for this model was formulated:

\begin{theorem} \label{thm_fil_sv_hmm}
In the HMM with non-constant volatility the unnormalized filter $\rho$ satisfies 
\[
\d\rho_t=Q^T\rho_t\d t+\frac{1}{\sigma_t^2}\diag(\mu)\rho_t\d R_t.
\]
\end{theorem}

Using that $\d R_t=\sigma_t\d\tilde{W}_t$ we see that
\[
\d\rho_t=Q^T\rho_t\d t+\diag(\mu)\rho_t\sigma_t^{-1}\d \tilde{W}_t,
\]
so the unnormalized filter itself is a diffusion with respect to $\tilde{W}$ and $\sigma_t$ can be chosen to be a function of the filter.

The intuition behind this choice for $\xi$ is that one interpretation of volatility is that it arises due to the accumulated actions of traders, and traders can only act on the information available to them. The information that is available to market participants about the state of the economy is encoded in the filter, and thus it makes sense to model the volatility as depending on this filter.

The volatility process $\sigma_t=\bar{\sigma}(\xi_t)$ is modeled as a function of an observable diffusion. This model allows for a wide range of possible volatility processes, incorporating e.g. level dependency via $\sigma_t=\bar{\sigma}(S_t)$ and dependencies on other $\F^R$-adapted processes that capture information about the state of the economy. From the form of the filtering equation in this setting we can see that we can also choose the unnormalized filter as the diffusion, since $\bar{\sigma}$ is assumed to be continuously differentiable with bounded derivatives. This leads to the question of how to choose $\xi$ and $\bar{\sigma}$ in a good way.\\

\subsection{The Optimal Approximating Model Given the Filtration}

We want a model for the stock prices that is as close as possible to the MSM while still allowing for a filtering problem. As we are modeling a return process, closeness here means that the output from both models, the return processes, are close to each other in a quantifiable way. A typical measure is the expected squared distance, also known in statistics as the \textit{mean-squared error} (MSE). Since the MSE is given by an expectation it is not always easy to investigate or clear how to decrease it.\\
Consider a continuous-time MSM with parameters $\mu,\sigma^\msm$ and $Q$. The return process $R^{\msm}$ is
\[
R_t^\msm=\int_0^t\mu^TY_s\d s+\int_0^t(\sigma^\msm)^TY_s\d W_s.
\]
Assume that we have some filtration $\G$ that describes the available information. We then want to choose a regime-switching model $R^\G$ driven by the same $Y,W$ and with observable volatility process $\sigma_t^\G$ that approximates the MSM-returns $R^{\msm}$. The best approximation would be the model $\tilde{R}$ with volatility $\tilde{\sigma}_t$ which minimizes the MSE for all times $t$, i.e.\
\[
\tilde{\sigma_t}=\arg\min_{\sigma_t^\G \G_t\text{-mbl.}}\E\left[\abs{R_t^\msm-R_t^\G}^2\right].
\]
Such a minimization problem is typically not easy to solve. Luckily it turns out that in our situation we can reduce the problem to the minimization problem that is solved by conditional expectations:
\begin{theorem}\label{thm_approx_it}
For a fixed filtration $\G$ the minimizer of 
\[
\min_{\sigma_t^\G \G_t\text{-mbl.}}\E\left[\abs{R_t^\msm-R_t^\G}^2\right]
\]
is given by
\[
\tilde{\sigma}_t=\sigma^T\E[Y_t\vert\G_t].
\]
\end{theorem}
\begin{proof}
The difference between the return process $R^\msm$ and any $R^\G$ is
\begin{align*}
R^\msm_t-R^\G_t&=\int_0^t\mu^TY_s\d s+\int_0^t\sigma^TY_s\d W_s-\int_0^t\mu^TY_s\d s-\int_0^t\sigma^\G_s\d W_s\\
&=\int_0^t\sigma^TY_s-\sigma_s^\G\d W_s,
\end{align*}
so with It\^{o}'s isometry it follows that
\[
\E[(R^\msm_t-R^\G_t)^2]=\E\Big[\big(\int_0^t\sigma^TY_s-\sigma_s^\G\d W_s\big)^2\Big]=\E\Big[\int_0^t(\sigma^TY_s-\sigma_s^\G)^2\d s\Big].
\]
With Fubini we can reformulate this as
\[
\E[(R^\msm_t-R^\G_t)^2]=\int_0^t\E[(\sigma^TY_s-\sigma_s^\G)^2]\d s.
\]
The integral is surely minimized if we can find a $\G$-measurable $\tilde{\sigma}$ such that at any time $s$ $\E[(\sigma^TY_s-\tilde{\sigma}_s)^2]$ is smaller than $\E[(\sigma^TY_s-\sigma_s^\G)^2]$  for any possible $\sigma_s^\G$.\\
But we know that minimizing the $L^2$-distance over the set of random variables measurable w.r.t.\ a certain $\sigma$-algebra is solved by the conditional expectation w.r.t.\ this $\sigma$-algebra, so
\[
\arg\min_{\sigma_s^\G\G_s\text{-mbl.}}\E[(\sigma^TY_s-\sigma_s^\G)^2]=\E[\sigma^TY_s\vert\G_s]
\]
and
\[
\E[\sigma^TY_s\vert\G_s]=\sigma^T\E[Y_s\vert\G_s].
\]
So the choice of $\tilde{\sigma}_t=\sigma^T\E[Y_t\vert\G_t]$ minimizes $\E[(R^\msm_t-R^\G_t)^2]$ over all processes that are $\G$-adapted.
\end{proof}

\subsection{Iteration and Social Learning}
We see that for a given information filtration $\G$ the best approximation of the MSM would be given by
\[
\tilde{R}_t=\int_0^t\mu^TY_s\d s+\int_0^t\sigma^T\E[Y_s\vert\G_s]\d W_s.
\]

But in practice one chooses $\G_t=\F_t^R$ which depends on the ``real'' volatility process. With this choice of $\G$ changing the model for the volatility changes the available information to $\G^1=\F^{\tilde{R}}$. We end up with a new minimization problem
\[
\tilde{\sigma_t}^{(2)}=\arg\min_{\sigma_t^\G \G_t^1\text{-mbl.}}\E\left[\abs{R_t^\msm-R_t^{\G^1}}^2\right]
\]
that is again solved by
\[
\tilde{\sigma_t}^{(2)}=\sigma^T\E[Y_t\vert\G_t^1].
\]
This gives rise to another filtration generated by $\d \tilde{R}^{(2)}_t= \mu^TY_t\d t+\tilde{\sigma}^{(2)}_t\d W_t$, $\G^2=\F^{\tilde{R}^{(2)}}$. For this filtration we can solve for the best approximation of the MSM using the information $\G^2$ and end up with another approximating model $\tilde{R}^{(3)}$ and a new information filtration.\\
In such a manner we can iteratively construct a sequence of models that all approximate the MSM and depend on different information filtrations. These models all have the serious drawback that their volatilities $\tilde{\sigma}^{(i)}$ are not adapted to the filtrations $\G^i$ generated by the corresponding returns. By construction the volatilities are instead adapted to the filtration $\G^{i-1}$ generated by the returns from the previous step. Thus, we end up with models that are not in the class of HMMs from \cite{Sass2004/2}. In particular, the results from \cite{Sass2004/2} cannot be applied to calculate their filters.\\
\begin{remark}[Bayesian Social Learning]\label{rem_social_learning}
 The above iterative model involving the filtered estimates can also be motivated from a social learning perspective, such models are used widely in economics
to model multi-agent systems see \cite{Cha04,Krishnamurthy2014}. 

Let us first describe the classical social learning model.
Consider a Markov chain  $Y$ with discrete time observations
$$R_{t_k}  = \mu^T Y_{t_k}  + \sigma_{t_k} w_{t_k} , \quad k = 1,2,\ldots $$
where $w$ denotes an independent and identically distributed process.
The dynamics of the social learning model proceeds iteratively as follows:
At time $t_k$, the public belief $ \hat{Y}_{k-1} = \E[ Y_{t_{k-1}} | a_1,\ldots, a_{k-1} ] $ is available.
At time $t_k$, given observation $R_{t_k}$, agent $k$ chooses  action
$$ a_k = \arg\min_{a \in \mathcal{A} }  \E[ c(Y_{t_k},a) | a_1,\ldots,a_{k-1},R_k ]  $$
for some specified cost function $c(y,a)$ and finite action set $\mathcal{A}$.
Given this new action $a_k$, the public belief is updated as
$$ \hat{Y}_k = \E[ Y_{t_k} | a_1,\ldots, a_k ] ,$$
and the model proceeds to time $t_{k+1}$, and so on.
An interesting property of the above social learning model is that agents can form herds and information cascades \cite{Cha04}.

We now construct a modified version of the model considered in this paper which is similar in spirit to the above social learning model.
Suppose
the first trader observes the returns, filters and trades accordingly yielding a new volatility since volatility is generated by trading. The next trader sees the returns coming from this new volatility, computes the filter (which may now be different) and trades accordingly, yielding again an new volatility. This iteration continues and might end at some model as defined below (if we reach a fixed point of this iteration). In practice such a transition of information would occur with small time-steps. Here the volatility is rather motivated as arising from an instantaneous updating which can be seen as instantaneous or degenerate social learning. For trading it might further be motivated from high-frequency trading which allows for a lot of iterations within short 
time intervals. 
\end{remark}
Considering the iterative construction of the models we potentially end up with an infinite sequence. It is not clear whether the iteration ever converges.\footnote{In the discrete time social learning model outlined above, assuming a  finite action set and finite observation set with the Markov chain having identity transition matrix, an elementary application of the martingale convergence theorem
shows that the public belief will converge to a constant with probability one in finite time - this is called an information cascade.}

And for it to truly stop we would somehow need a volatility process such that $\G^{i-1}$ and $\G^i$ coincide. The sequence obviously depends strongly on the initial information $\G$, thus the question is also how to choose it in a reasonable way. Considering our original goal of finding a regime-switching model with better properties compared to the HMM it makes sense to first start with information coming from HMM returns, $\G=\F^\hmm$. Then in the next step we have
\[
\tilde{\sigma}_t=\sigma^T\hy_t^\hmm,
\]
where $\hy^\hmm_t=\E[Y_t\vert\F_t^\hmm]$ is the HMM-filter. But we still have the problem that in this model we cannot compute the filter. We can think about this model as the first step on the way from an HMM to an HMM with non-constant volatility, but we do not truly reach the class of models from \cite{Sass2004/2}.

\subsection{The Filterbased Volatility Model}
Thus, to make our approximation approach useful to really find a model where we can compute the filter, we now construct a model that is invariant under the iteration. Consider the following model
\[
\d R_t^\sv=\mu^TY_t\d t+\sigma^T\E[Y_t\vert\F^\sv_t]\d W_t,
\]
where $\F_t^\sv=\F_t^{R^\sv}$, i.e.\ the volatility is adapted to the observation filtration generated by the returns. This is now a coupling between the volatility and the observations as in \cite{Sass2004/2} and it guarantees that applying the iteration does not change the model, since then
\[
\tilde{\sigma}_t=\sigma^T\E[Y_t\vert\F^\sv_t].
\]
With this choice we now have an approximation of the MSM that is in this sense consistent and minimizes the MSE over all other volatility processes adapted to $\F^\sv$. This gives us a model that truly lies between the HMM and MSM. The distance to the MSM is given by
\[
\E\left[\abs{R_t^\msm-R_t^\sv}^2\right]=\int_0^t\E\left[(\sigma^TY_s-\sigma^T\E[Y_t\vert\F^\sv_t])^2\right]\d s.
\]

This result has deep implications for the volatility choice. In the formulation of the model we started with an arbitrary observable diffusion and only assumed the usual conditions on the functions in the diffusion to ensure the existence of a strong solution. Only by looking at the filtering equation we see that we can e.g.\ choose the unnormalized filter $\rho$ as this diffusion. At this point the motivation for this is only economic intuition since one typically interprets the volatility as arising because of traders whose actions have to be adapted to the available information. As we have pointed out before, the filter collects this information into the best estimate of the true state of the economy.\\
The iterative method for approximating the MSM now implies that using the filter as diffusion $\xi$ is not only an economically reasonable choice but the choice for approximating the MSM where the iteration stops. And furthermore, it also answers the next question after choosing the diffusion, that is how to choose the function $\sigma$ that connects the diffusion to the volatility process. It suggests that the observable volatility process that leads to the best approximation of the MSM is the function with which the volatility in the MSM depends on the chain now applied to the best estimate of the chain, that is, the filter \hy.\\
The function $\bar{\sigma}$ is the scalar product of the vector $\sigma$ with the normalized filter. As such it has the nice property that it is linear in \hy. That does not seem very surprising looking at the simple volatility structure of the MSM we want to approximate, but considering the structure of the HMM in \cite{Sass2004/2} this is a very useful result. This HMM in general allows for much more complicated volatility functions, but Theorem \ref{thm_approx_it} tells us that already the simple scalar product of the filter with the volatility vector from the MSM is the best approximation.\\
Now, after pointing out these very advantageous practical implications, we have to pay attention to the theoretical intricacies arising in this model for the volatility: The volatility in the HMM from \cite{Sass2004/2} depends on a \tw-diffusion, and the filtering equation implies that the unnormalized filter is such a diffusion. Theorem \ref{thm_approx_it} suggests the volatility as a function of the \textit{normalized} filter, so the model actually reads
\[
\d R_t=\mu^TY_t\d t+\sigma^T\frac{\rho_t}{1^T\rho_t} \d W_t.
\]

But the function $\sigma^T\frac{x}{1^Tx}$ does no have bounded derivatives w.r.t\ the $x_i$, so the Zakai equation proven in \cite{Sass2004/2} cannot be applied. The model is still well-defined, since for $x$ with strictly positive entries the function $\sigma^T\frac{x}{1^Tx}$ is not only bounded, but also bounded away from $0$:
\[
0<\min_i \sigma_i\leq \sigma^T\frac{x}{1^Tx}\leq\max_i\sigma_i
\]
as $\sigma$ has only positive entries. This implies that Novikov's condition is fulfilled and we can do the usual measure change leading to the unnormalized filter $\rho_t$. With the existence of the unnormalized filter ensured we see that the model with $\sigma_t=\sigma^T\frac{\rho_t}{1^T\rho_t}$ is indeed well-defined.
This leads to the introduction of a new model:
\begin{definition}
Let the return process $R$ have the dynamics
\[
\d R_t=\mu^TY_t\d t+\sigma^T\hy_t\d W_t
\]
with $Y$ an irreducible, time-homogenous Markov chain with state space $\{e_1,\dots,e_d\}$, $\mu\in\R^d$, $\sigma\in\R^d_{>0}$ and $\hy_t=\E[Y_t\vert\F_t^R]$.
This model is then called the Filterbased-Volatility Hidden Markov Model (FB-HMM).
\end{definition}

As we have already pointed out, this looks like a subclass of the HMM since also
\[
\d R_t=\mu^TY_t\d t+\sigma^T\frac{\rho_t}{1^T\rho_t} \d W_t,
\]
but as we cannot apply Theorem \ref{thm_fil_sv_hmm} it is not clear whether $\rho_t$ is truly a \tw-diffusion. One of the major advantages of regime-switching models is that their filters are actually given by a finite system of equations. So to make our model tractable we have to derive the dynamics of the filter, which we can be done following the methods from \cite{Elliott1993}.

\begin{theorem}
In the Filterbased Volatility Model the unnormalized filter $\rho$ satisfies
\[
\d \rho_t=Q^T\rho_t\d t+\sigma_t^{-2}\diag(\mu)\rho_t\d R_t.
\]
\end{theorem}
\begin{proof}
The function $\pinner{\sigma}{\frac{\rho_t}{\id^T\rho_t}}$ is bounded from above in $\rho_t$ and also bounded away from $0$. Using these properties the claim can be proven using similar methods as in \cite{Elliott1993}.
\end{proof}

\subsection{Stylized Facts}

Real data often exhibit certain characteristics, known as \textit{stylized facts}, see e.g.\ \cite{Ryden1998},\cite{Cont2001}. This behavior is common across a wide range of instruments, markets and time periods. Thus an approach for assessing the quality of a model is to see which stylized facts can be reproduced by a return process coming from the model. The main motivation for the introduction of the FB-HMM were the better econometric properties of the MSM due to its switching volatility. The MSM can e.g.\ capture volatility clustering, which is the observation that high volatility events tend to cluster in time, followed by phases of calmer movement of the prices.
\subparagraph*{Volatility Clustering}
In the FB-HMM the volatility process is controlled by the filter. So any clustering of high volatility events must come from time periods where the filter puts a comparably large probability weight on states where the corresponding entry in $\sigma$ is higher than the others. The filter is an estimate of the true state of $Y$, so if $Y$ is in the state with large entry in $\sigma$ and the filter performs well, then there will be a phase of high volatility. Volatility clustering in the filterbased model is thus connected to not only the values in $\sigma$ directly, as it is in the MSM via the Markov chain, but also to the performance of the filter.
\subparagraph*{Leverage Effect} The so-called leverage effect concerns an asymmetric relation between the volatility and the returns of a financial asset. In particular, a statistical leverage effect is the fact that the correlation between past returns and future volatility is mostly observed to be negative, which is in accord with one's economic intuition. In the MSM one can derive a heuristic for choosing the parameters that leads to a negative value of the approximated covariance. Due to the tractability of the transition probabilities in the 2-state model, one can then prove that in this case the heuristic leads to a negative covariance, see \cite{Leoff2016}. Since the FB-HMM approximates the MSM, it is intuitive to apply the same heuristic also here. But the coupled structure makes it impossible to analytically derive the covariance even for $d=2$. Numerical examples still indicate a negative covariance, but just as for volatility clustering, the extent to what this stylized fact can be introduced depends directly on the performance of the filter.
\subparagraph*{Asymmetry} The next stylized fact which we consider is the often observed asymmetry of the return distribution. Large drawdowns in prices are more often observed than equally large upward movements. This asymmetry arises naturally e.g.\ from the existence of thresholds below which positions must be cut unconditionally for regulatory reasons. Declining stock prices are more likely to give rise to massive portfolio rebalancing (and thus volatility) than increasing stock prices, which connects the asymmetry of return distributions to the leverage effect.\\
Cont noted in \cite[Section 4.1]{Cont2001} that one must use more than only the standard deviation to capture the variability of the returns due to the non-Gaussian nature of most of the empirical data. One popular approach is to consider higher moments as a measure of dispersion.\\
These moments are very involved in the continuous-time MSM due to the unbounded number of possible jumps in every finite interval. Instead we can again employ a heuristic choice that uses the different entries in $\sigma$ to push more of the probability mass to the left tail of the distribution. This choice can be made in a way that is consistent with the heuristic for the leverage effect. We can apply the same choice for the FB-HMM and underline with numerical examples in Section \ref{sec_numerics} that these parameters then indeed lead to an asymmetry in the histograms.\\
\subparagraph*{Autocorrelation Function}
The autocorrelation functions (acf) are used to describe and quantify the long-term dependence between increments of the return processes. Their properties are similar over a wide range of  empirical data and thus for model-building it is of importance to know their characteristics. In particular one is interested in the linear autocorrelations and the autocorrelations of the absolute value of the returns.\\
In \cite{Fruehwirth2006} the autocorrelations of the discrete-time HMM and MSM are presented following the work of Timmermann in \cite{Timmermann2000}. In \cite{Leoff2016} we applied similar methods for the continuous-time case. For the FB-HMM we have the following lemma:
\begin{lemma}\label{lemma_acf_sv}
If $R$ follows the FB-HMM we have for $0\leq t<s\leq T$ with $t+\Delta_t< s$
\begin{align*}
\mathrm{Cov}(\Delta R_t, \Delta R_s)&=\int_s^{s+\Delta_t}\int_t^{t+\Delta_t}\sum_{i,j=1}^d\mu_i\mu_j(\e^{Q(r-u)})_{ij} \nu_i\d u\d r\\ 
&\quad +\E\Big[\int_t^{t+\Delta_t}\sigma^T\hat{Y}_r\d W_r\int_s^{s+\Delta_t}\mu^TY_r\d r\Big]-\Delta_t^2(\mu^T\nu)^2.
\end{align*}
\end{lemma}

\begin{proof}
Calculating the covariance and using $\E[\hy_t]=\E[Y_t]=\nu$ implies
\begin{align*}
&\mathrm{Cov}(\Delta R_t, \Delta R_s)\\
&=\E\Big[\Big(\int_t^{t+\Delta_t}\mu^TY_r\d r+\int_t^{t+\Delta_t}\sigma^T\hy_r\d W_r -\Delta_t\mu^T\nu\Big)\\
&\quad\Big(\int_s^{s+\Delta_t}\mu^TY_r\d r+\int_s^{s+\Delta_t}\sigma^T\hy_r\d W_r-\Delta_t\mu^T\nu\Big)\Big]\\
&= \E\Big[\int_t^{t+\Delta_t}\mu^TY_r\d r \int_s^{s+\Delta_t}\mu^TY_r\d r\Big]\\
&\quad +\E\Big[\int_t^{t+\Delta_t}\mu^TY_r\d r\int_s^{s+\Delta_t}\sigma^T\hy_r\d W_r+\int_t^{t+\Delta_t}\sigma^T\hy_r\d W_r\int_s^{s+\Delta_t}\mu^TY_r\d r\Big]\\
&\quad +\E\Big[\int_t^{t+\Delta_t}\sigma^T\hy_r\d W_r\int_s^{s+\Delta_t}\sigma^T\hy_r\d W_r\Big] -(\Delta_t\mu^T\nu)^2.
\end{align*}

The first summand depends on $Y$ which has only finitely many states, thus we can reformulate it as
\[
\E\Big[\int_t^{t+\Delta_t}\mu^TY_r\d r \int_s^{s+\Delta_t}\mu^TY_r\d r\Big]=\int_s^{s+\Delta_t}\int_t^{t+\Delta_t}\sum_{i,j=1}^d\mu_i\mu_j(\e^{Q(r-u)})_{ij} \nu_i\d u\d r.
\]

For the last expectation we apply It\^{o}'s isometry to obtain
\begin{align*}
\E\Big[\int_t^{t+\Delta_t}\hspace{-4pt}\sigma^T\hy_r\d W_r\int_s^{s+\Delta_t}\hspace{-4pt}\sigma^T\hy_r\d W_r\Big]&=\E\Big[\int_0^{s+\Delta_t}\hspace{-4pt}\mathbbm{1}_{[t,t+\Delta_t]}(r)\sigma^T\hy_r \mathbbm{1}_{[s,s+\Delta_t]}\sigma^T\hy_r \d r \Big]\\
&=0.
\end{align*}

Now it is left to show that $\E\Big[\int_t^{t+\Delta_t}\mu^TY_r\d r\int_s^{s+\Delta_t}\sigma^T\hy_r\d W_r\Big]=0$. Observe that since $\hat{Y}$ is bounded $\int_0^u\sigma^T\hy_r\d W_r $ is a true martingale w.r.t.\ $\F_u^{Y,W}$, so 
\begin{align*}
\E\Big[\int_s^{s+\Delta_t}\hspace{-4pt}\sigma^T\hy_r\d W_r\vert\F_s^{Y,W}\Big]&=\E\Big[\int_0^{s+\Delta_t}\hspace{-4pt}\sigma^T\hy_r\d W_r\vert\F_s^{Y,W}\Big]-\E\Big[\int_0^s\hspace{-4pt}\sigma^T\hy_r\d W_r\vert\F_s^{Y,W}\Big]\\
&=\int_0^s\sigma^T\hy_r\d W_r-\int_0^s\sigma^T\hy_r\d W_r\\
&=0
\end{align*}
and since $t+\Delta_t<s$
\begin{align*}
\E\Big[\int_t^{t+\Delta_t}\mu^TY_r\d &r\int_s^{s+\Delta_t}\sigma^T\hy_r\d W_r\Big]\\
&=\E\Big[\E\big[\int_t^{t+\Delta_t}\mu^TY_r\d r\int_s^{s+\Delta_t}\sigma^T\hy_r\d W_r\vert\F_s^{Y,W}\big]\Big]\\
&=\E\Big[\int_t^{t+\Delta_t}\mu^TY_r\d r\E\big[\int_s^{s+\Delta_t}\sigma^T\hy_r\d W_r\vert\F_s^{Y,W}\big]\Big]\\
&=0. 
\end{align*}
\end{proof}

This representation of the linear autocovariances can be used for numerical estimation. In \cite{Leoff2016} we also estimated the autocorrelations of the absolute value, the sign and the squares of the return increments. We saw that they seem to converge to $0$, but slower than in both the MSM and HMM. For the linear acf this is probably due to the additional summand $\E[\int_t^{t+\Delta_t}\sigma^T\hat{Y}_r\d W_r\int_s^{s+\Delta_t}\mu^TY_r\d r]$, and we suspect a similar effect for the other acf. The autocorrelations of the squares were significantly positive in the simulations, which underlines the fact that there is volatility clustering in the FB-HMM. Comparing the estimations for the HMM and MSM to the FB-HMM we could make another interesting observation: In the MSM the linear acf has the same order of magnitude as the absolute acf, while in the HMM the linear acf is much larger than the absolute acf. In the FB-HMM we observed that as in the HMM the linear acf is larger than the absolute acf, but the effect is much less pronounced. This again supports the intuition with which we introduced the FB-HMM: In a sense the FB-HMM lies between the MSM and the HMM.

\section{Consistency of the Discretized Model}\label{sec_consistency}
As we are interested in the relationship between discrete- and continuous-time models, we will now examine a notion of consistency between the discretized and the continous-time returns in our new model. 

The FB-HMM in discrete time can be defined in a similar way as a discrete-time HMM as
\[
R_k=\mu^TY_{k-1}+\sigma^T\hy_{k-1}w_k,
\]
with the $w_k$ i.i.d.\ $\mathcal{N}(0,1)$-distributed.\\
If we want to obtain a model that truly lives in discrete time $\hy_{k-1}$ has to be the estimate of $Y_{k-1}$ at time $k-1$ using only the discrete-time returns, i.e.\ $\hy_{k-1}=\E[Y_{k-1}\vert\F^R_{k-1}]$.\\
Instead of this discrete-time model we can also discretize the returns $R$ in the continuous-time model directly as
\begin{equation} \label{equ_disc_fb_hmm}
\bar{R}_k=\mu^TY_{t_{k-1}}\Delta_t+\sigma^T\hy_{t_{k-1}}(W_{t_k}-W_{t_{k-1}}),
\end{equation}
where we now use the discretized continuous-time filter $\hy_{t_{k-1}}=\E[Y_{t_{k-1}}\vert\F^R_{t_{k-1}}]$. In the HMM and MSM these discretizations would be the same: In the HMM we would use $\sigma^T\nu$ for both $R_k$ and $\bar{R}_k$, and in the MSM we would directly use the chain $Y$ at time $t_{k-1}$. This kind of consistency is not given for the FB-HMM.\\
Discretizing the continuous-time returns directly does not lead to a time series that follows a discrete-time FB-HMM. Instead we consider the consistency in the returns themselves by investigating the behavior of the discretized returns for a discretization step converging to $0$. In the following we will prove that if we aggregate the $\bar{R}_k$ over time, the resulting sum does indeed converge to the continuous-time process.\\

Let $R$ follow a continuous-time FB-HMM with parameters $\mu,\sigma$ and $Q$
\[
\d R_t=\mu^TY_t\d t+\sigma^T\hy_t\d W_t.
\]
For $n\in\N$ consider the discretization $0=t_0^n,t_1^n,\dots, t_n^n=T$ of the time interval $[0,T]$ given by
\[
\Delta_t^n=\frac{T}{n}, \qquad t_k^n=k\Delta_t^n
\]
and define the discretization of $R$ as
\[
R_k^n=R_{k-1}^n+\mu^TY_{t^n_{k-1}}\Delta_t^n+\sigma^T\hy_{k-1}^n\Delta W^n_k, \qquad R_0^n=0
\]
where $\hy^n_k,Y_k^n$ and $\Delta W^n_k$ are now the discretizations of $\hy,Y$ and $W_{t^n_k}-W_{t^n_{k-1}}$ respectively, i.e.
\begin{align*}
Y_k^n&=Y_{t_k^n},\\
\hy_k^n&=\hy_{t_k^n},\\
\Delta W_k^n&=W_{t_k^n}-W_{t_{k-1}^n}.
\end{align*}
This discretization $R^n$ is just the cumulative sum of the usual discretizations. We use the sum here since we want to approximate a continuous-time process that is given by an integral. This kind of method is usually known as the (explicit) \textit{Euler scheme}. It has the advantage that it is recursive and thus easy to implement.\\
To quantify the global discretization error we introduce the piecewise constant approximations of $Y,\hy$ on this grid as
\begin{align*}
Y^n_t&=Y_{t_k^n},\\
 \hy^n_t&=\hy_{t_k^n}
\end{align*}
for $t \in [t_k^n,t^n_{k+1})$.
Then the difference between the discretization $R^n$ and the continuous-time process at the final time is given by
\[
\epsilon^n\colonequals R_T-R_n^n=\int_0^T\mu^TY_t\d t+\int_0^T\sigma^T\hy_t-\sum_{k=1}^n\mu^TY_{k-1}^n\Delta_t^n++\sigma^T\hy_{k-1}^n\Delta W_k^n.
\]
Using $Y^n$ and $\hy^n$ this can be reformulated as
\[
\epsilon^n=\int_0^T\mu^T(Y_t-Y_t^n)\d t+\int_0^T\sigma^T(\hy_t-\hy_t^n)\d W_t.
\]

We now show that this global error actually converges to zero if the grid becomes finer. This means that already the simple Euler discretization of the continuous-time returns is convergent and there is no need for higher-order methods to achieve consistency.

\begin{theorem} \label{thm_fb_returns_euler}
For arbitrarily fine grids the global discretization error at $T$ tends to 0 in $L^2$, that is
\[
\E\big[\norm{\epsilon^n}^2\big]\to 0, \quad(n\to\infty).
\]
\end{theorem}
\begin{proof}
First we gather some observations that we will use in the course of the proof:
\begin{enumerate}[i)]
\item For $a,b\in \R$ 
\begin{align*}
(a+b)^2&\leq (\vert a\vert+\vert b\vert)^2\\
(\vert a\vert+\vert b\vert)^2&\leq 4( a^2+ b^2).
\end{align*}
\item For any real valued function $f$
\[
\Big( \int_0^tf(s)\d s\Big)^2\leq \Big(\int_0^t\vert f(s)\vert\d s \Big)^2.
\]
\item For $N\in\N$ and arbitrary $a_k\in\R$ Jensen's inequality implies
\[
\Big(\frac{1}{N}\sum_{k=1}^N{a_k}\Big)^2\leq \frac{1}{N}\sum_{k=1}^N{a_k^2}.
\]
\end{enumerate}

Now we proceed with proving the theorem: The squared $L^2$-norm of the error $\epsilon^n$ can be estimated using the reformulation $\epsilon^n=\int_0^T\mu^T(Y_t-Y_t^n)\d t+\int_0^T\sigma^T(\hy_t-\hy_t^n)\d W_t$ as
\begin{align*}
\norm{\epsilon^n}^2_{L^2}&=\E\Big[ \Big( \int_0^T\mu^T(Y_t-Y_t^n)\d t+\int_0^T\sigma^T(\hy_t-\hy_t^n)\d W_t \Big)^2 \Big]\\
&\stackrel{\text{i)}}{\leq} 
\E\Big[ \Big( \big\vert\int_0^T\mu^T(Y_t-Y_t^n)\d t\big\vert+\big\vert\int_0^T\sigma^T(\hy_t-\hy_t^n)\d W_t\big\vert\Big)^2  \Big]\\
&\stackrel{\text{i)}}{\leq}
4\E\Big[ \Big( \int_0^T\mu^T(Y_t-Y_t^n)\d t \Big)^2 \Big]+4\E\Big[ \Big( \int_0^T\sigma^T(\hy_t-\hy_t^n)\d W_t\Big)^2  \Big].\\
\end{align*}
Thus it suffices to show that both of the summands converge to $0$.\\
First, consider that 
\begin{equation*}
\E\Big[ \Big( \int_0^T\mu^T(Y_t-Y_t^n)\d t \Big)^2 \Big]\stackrel{\text{ii)}}{\leq}\E\Big[ \Big( \int_0^T\vert\mu^T(Y_t-Y_t^n)\d t \Big)^2 \Big]
\end{equation*}
and by Cauchy-Schwartz
\begin{align*}
\int_0^T\vert\mu^T(Y_t-Y_t^n)\vert\d t&=\sum_{k=1}^n{\mathbbm{1}_{\{Y\text{ jumps on }[t^n_{k-1},t_k)\}}\int_{t^n_{k-1}}^{t^n_k}\vert\mu^T(Y_t-Y_t^n)\vert\d t} \\
&\leq 2\norm{\mu}\Delta_t^n \sum_{k=1}^n{\mathbbm{1}_{\{Y\text{ jumps on }[t^n_{k-1},t^n_k)\}}}.
\end{align*}
So 
\begin{align*}
\E\Big[ \Big( \int_0^T\mu^T(Y_t-Y_t^n)\d t \Big)^2 \Big]&\leq \E\Big[ \Big( 2\norm{\mu}\Delta_t^n \sum_{k=1}^n{\mathbbm{1}_{\{Y\text{ jumps on }[t^n_{k-1},t^n_k)\}}} \Big)^2 \Big]\\
&=4\norm{\mu}^2(\Delta_t^n)^2\E\Big[ \Big( \sum_{k=1}^n{\mathbbm{1}_{\{Y\text{ jumps on }[t^n_{k-1},t^n_k)\}}} \Big)^2 \Big]\\
&\stackrel{(\text{iii)}}{\leq}4\norm{\mu}^2(\Delta_t^n)^2n^2\E\Big[ \frac{1}{n}\sum_{k=1}^n{\mathbbm{1}_{\{Y\text{ jumps on }[t^n_{k-1},t^n_k)\}}} \Big]\\
&=4\norm{\mu}^2\frac{T^2}{n}\sum_{k=1}^n{\P(Y\text{ jumps on }[t^n_{k-1},t^n_k))}\\
&=4\norm{\mu}^2T^2\P(Y\text{ jumps on }[0,t^n_1)),\\
\end{align*}
and since $t_1^n\to 0$ for $n\to\infty$, $\P(Y\text{ jumps on }[0,t^n_1))\to 0$.

For the second summand we have with It\^{o} and Cauchy-Schwartz
\begin{align*}
\E\Big[ \Big( \int_0^T\sigma^T(\hy_t-\hy_t^n)\d W_t\Big)^2  \Big]&=\E\Big[ \int_0^T(\sigma^T(\hy_t-\hy_t^n))^2\d t \Big]\\
&\leq \E\Big[ \int_0^T\norm{\sigma}^2\norm{\hy_t-\hy_t^n}^2\d t  \Big]\\
&\leq \norm{\sigma}^2\E\Big[ \sum_{k=1}^n{\Delta_t^n \sup_{s\in[t^n_{k-1},t^n_k)}\norm{\hy_t-\hy_t^n}^2} \Big]\\
&\leq \norm{\sigma}^2\Delta_t^nn\E\Big[  \sup_{s\in[0,T)}\norm{\hy_t-\hy_t^n}^2 \Big]\\
&= \norm{\sigma}^2T\E\Big[  \sup_{s\in[0,T)}\norm{\hy_t-\hy_t^n}^2 \Big].
\end{align*}
This expectation converges to $0$, because $\sup_{s\in[0,T)}\norm{\hy_s-\hy_s^n}^2$ converges to $0$ almost surely.
\end{proof}
Theorem \ref{thm_fb_returns_euler} shows that the discretized returns from an FB-HMM are consistent with the continuous-time returns in the sense that in the limit where $\Delta_t\to 0$ the Euler discretizations converge in $L^2$ to the continuous-time returns. This is particularly beneficial since the explicit Euler scheme is very easy to implement.\\
\begin{remark}
It is important to note that Theorem \ref{thm_fb_returns_euler} does not deal with consistency of the \textit{filters}. As we have mentioned before, the discretization $R^n$ depends on the discretized filter $\hy_t^n$, not on the discrete-time filter, so $R^n$ does not follow a truly discrete-time model.\\
The Euler discretization of the HMM is $\bar{R}^\hmm_k=\mu^TY_{t_{k-1}}\Delta_t+\sigma_0(W_{t_k}-W_{t_{k-1}})$. Comparing this to \eqref{equ_disc_fb_hmm} we see the discretization of the HMM indeed follows a discrete-time HMM.
\end{remark}
\begin{remark}[Conclusion]
Following Haussmann and Sass \cite{Sass2004/2} optimal trading strategies $\pi^*$ can be computed in the FB-HMM and would for logarithmic utility $U_0(x) = \log(x)$ be of the form that in Section \ref{sec_po} we just replace $\sigma_0$ by $\sigma^T\hat Y_t$ in $\pi_t^\hmm$ or equivalently $Y_t$ by $\hat Y_t$ in $\pi_t^\msm$. This yields in both cases
\[
\pi_t^* = (\sigma^T \hat Y_t)^{-2} \mu^T\hat Y_t  
\]
which again illustrates that the FB-HMM lies between HMM and MSM. An extension to power utility is possible. The consistency shown in Theorem \ref{thm_fb_returns_euler} can then be used to show that optimization results in the continuous-time FB-HMM are close to the results obtained in the discretized model (for fractions within [0,1]) providing the problem-specific discretization consistency for the FB-HMM. Thus, while having on one hand better econometric properties than the HMM by approximating the MSM, it has on the other hand better comoputational and consistency properties than the MSM.
\end{remark}

\section{Numerical Examples}\label{sec_numerics}

\subsubsection*{HMM and MSM Sample Paths}
In Figure \ref{fig_path_msm_hmm} we give an example of the returns of MSM and HMM, depending on the same Markov chain and Brownian motion. One can clearly see the effect of the jumping volatility in the MSM and the resulting volatility clustering.\\
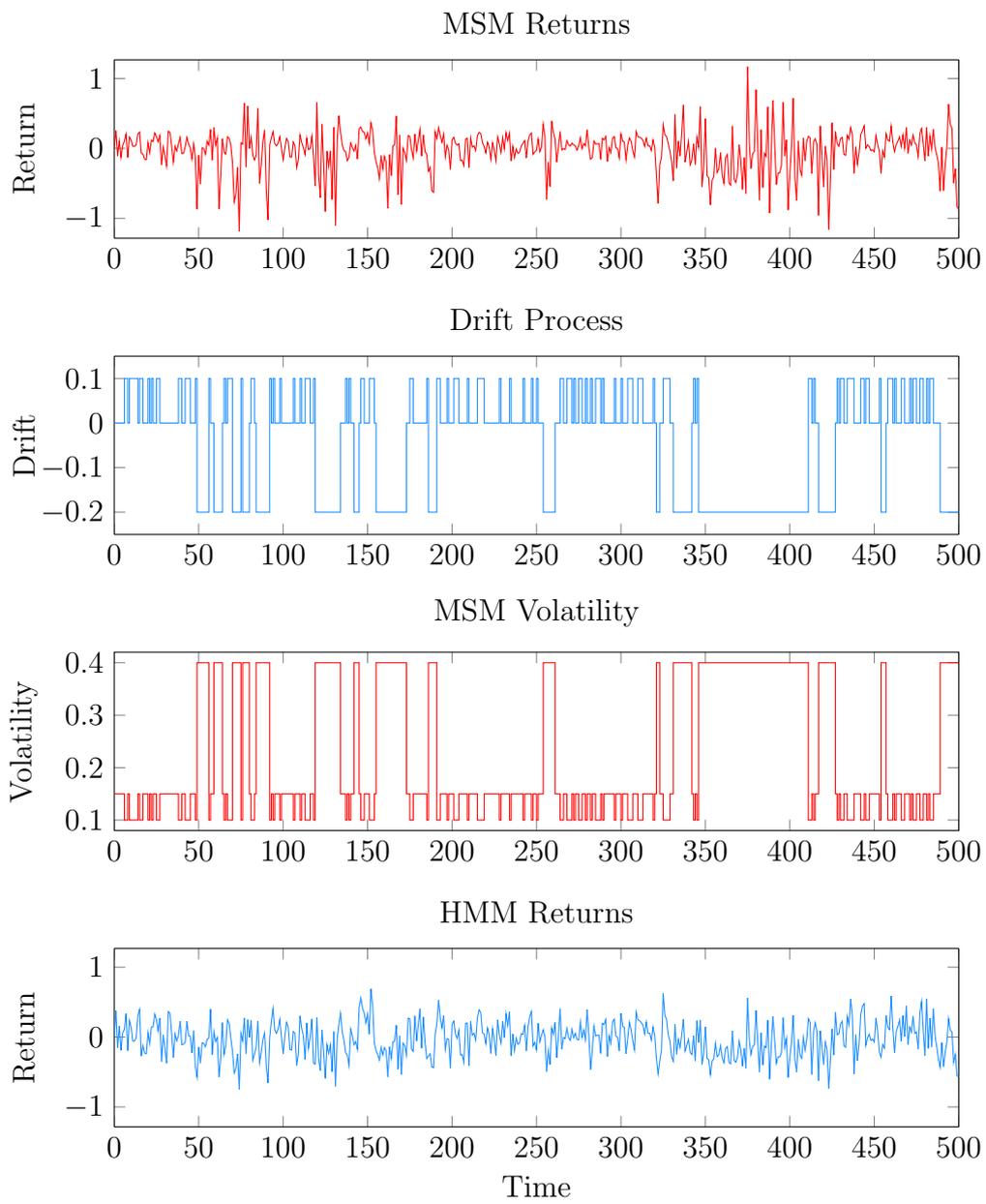
\begin{figure} 
\centering 
\input{fig_path_msm_hmm_2.tikz} 
\caption{Example of an HMM and MSM return series depending on the same Markov chain} 
\label{fig_path_msm_hmm} 
\end{figure}
\subsubsection*{Filtering in HMM and MSM}
In Proposition \ref{prop_no_filter_MSM} we have proven that the HMM and MSM behave differently during the transition to continuous time. This difference can already be seen by considering discretizations with very fine grids. To illustrate the different behavior for increasing grid size we consider an HMM and MSM depending on the same $Y$ and $W$, discretized with $n=10{,}000$ grid points over $T=1$ year. The example paths are presented in Figure \ref{fig_msm_10000} for the MSM and Figure \ref{fig_hmm_10000} for the HMM.
\begin{figure}
\begin{center}
\includegraphics[width=12cm]{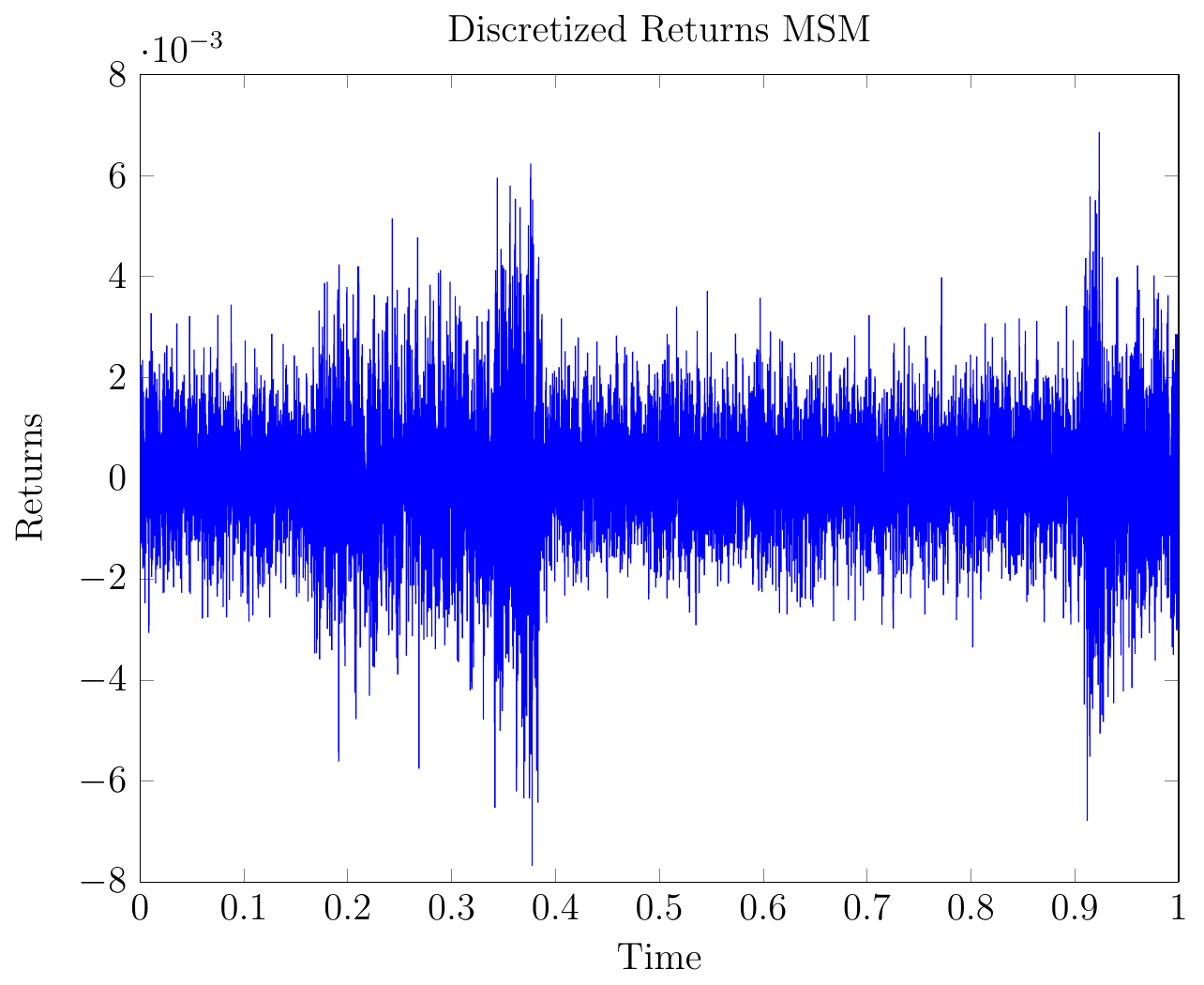} 
\includegraphics[width=12cm]{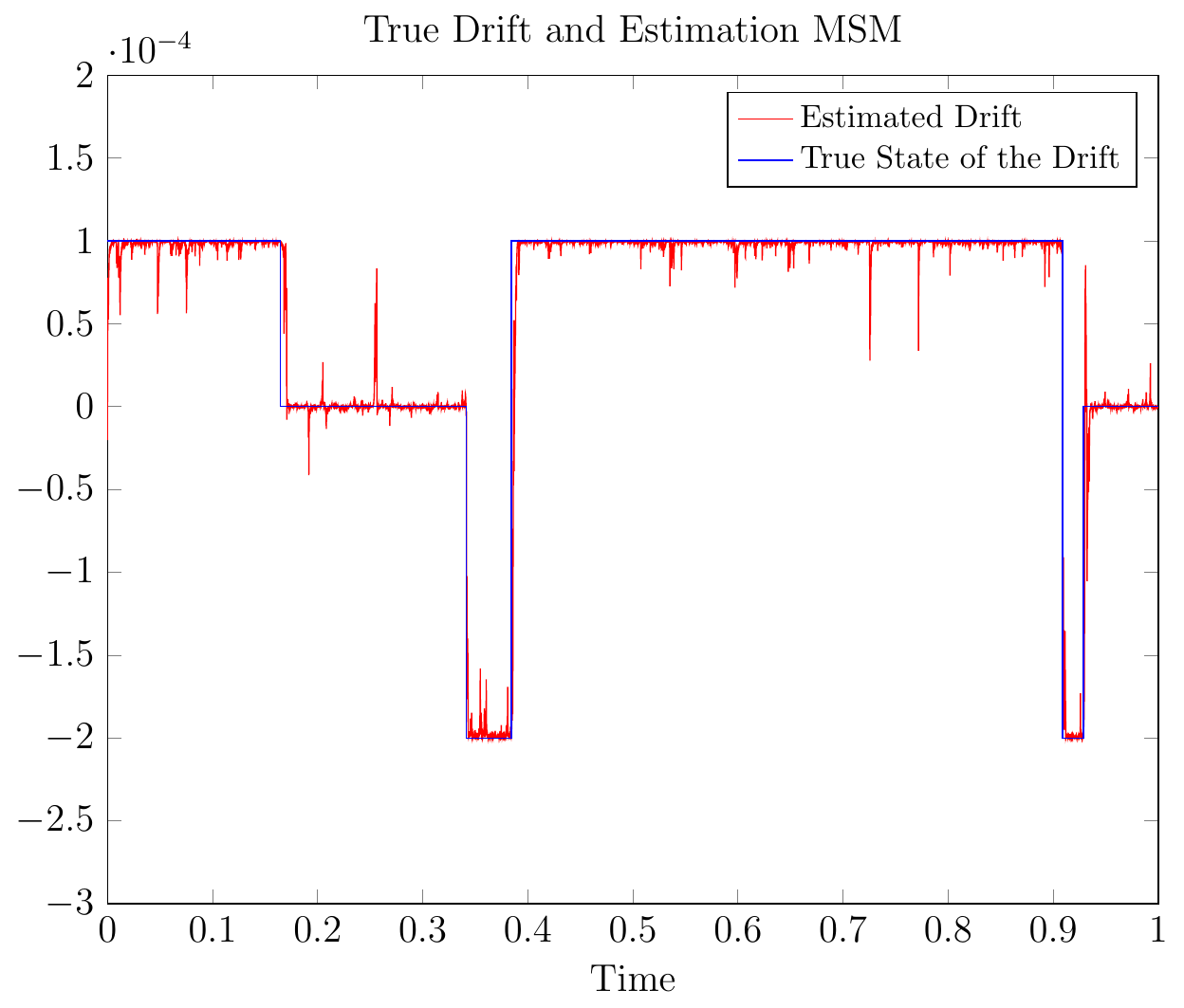} 
\caption{Example Path MSM with $10{,}000$ Grid Points} 
\label{fig_msm_10000} 
\end{center}
\end{figure}
\begin{figure}
\begin{center}
\includegraphics[width=12cm]{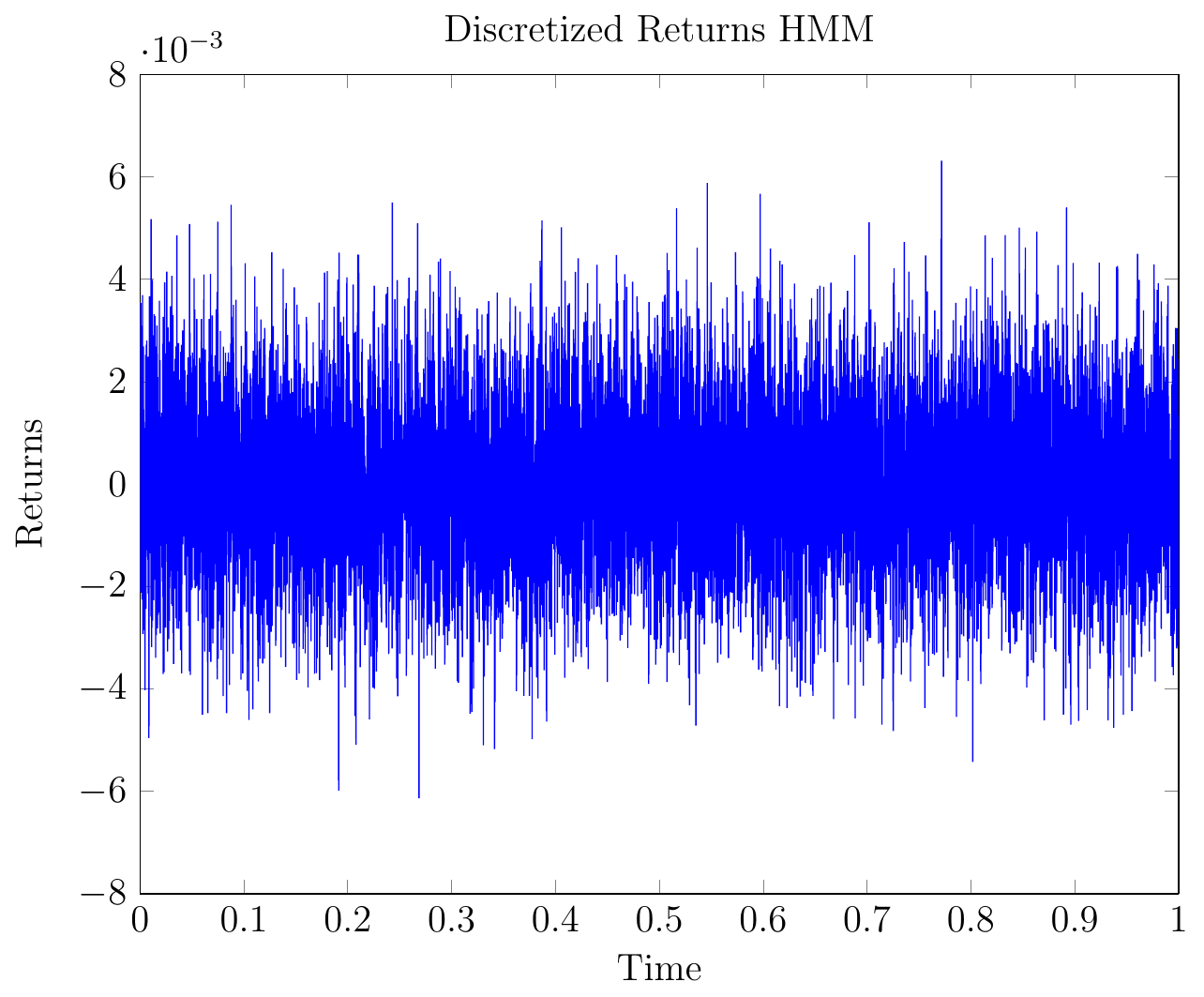} 
\includegraphics[width=12cm]{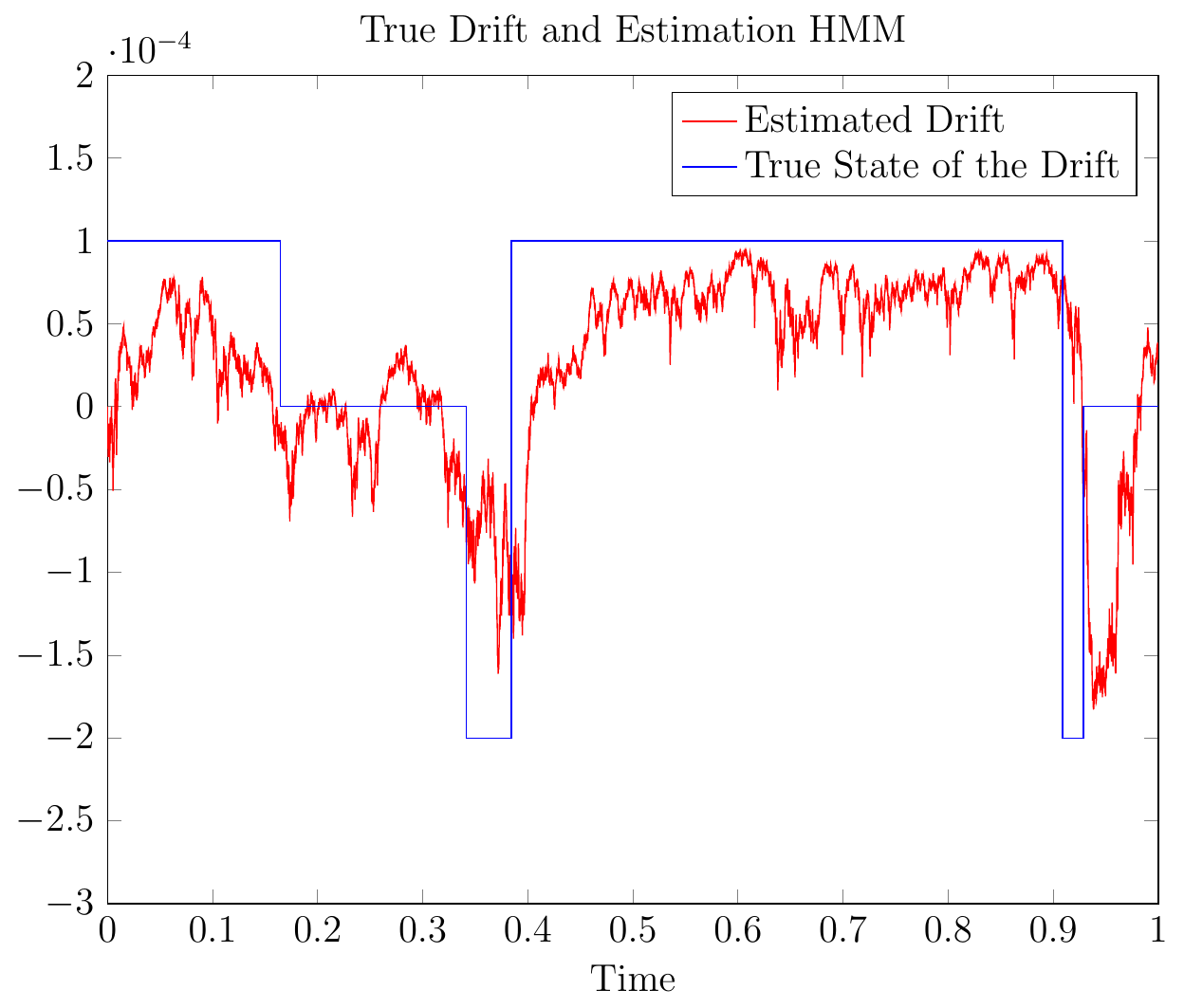} 
\caption{Example Path HMM with $10{,}000$ Grid Points} 
\label{fig_hmm_10000} 
\end{center}
\end{figure}
The difference even at this pathwise level is staggering. The performance of the MSM-filter improves significantly, while the HMM-filter does not. Looking at the returns we see intuitively why this is the case: With so many grid points we can strongly see the effect of the changing volatility in the MSM. This allows us to guess the state of the Markov chain quite well for a lot of grid points just by considering the amplitudes of the returns. Recalling the proof of Proposition \ref{prop_no_filter_MSM}, where we proved the observability of the signal in the continuous-time MSM via the quadratic variation, we can see the implications of this theoretical connection now even at the pathwise level. We can also observe the problems of trying to exploit this theoretical feature in practice: The chain can be estimated well only for an unrealistically large number of grid points. So this effect is not as relevant for e.g. daily data, which for $T=1$ leads to only 250 data points. But it is important and has a large impact for investigating the transition to continuous time.\\
Problems for the filter and the observer arise where e.g. quite extreme values for the normal distribution are simulated. These values imply a larger volatility than is really present and thus irritate the filter. But with a high probability these outliers do not persist over a longer time, since the Brownian increments are independent of each other. Thus over more time points the filter returns to a better estimate of the drift and does not really believe in a jump in the underlying chain.\\
In the HMM we do not have the help of a changing volatility. Looking at the returns and the performance of the filter we can see this clearly.\\

\subsubsection*{Filterbased HMM}
Due to the structure of the volatility in the FB-HMM it is necessary to compute a discretization of the filter $\hy_k$ for simulating discretized returns $R_k$. This is again different to the HMM or MSM, where the path of the filter is not necessary for simulating the returns.\\
We start with
\[
R_0=0,\qquad \rho_0=\nu, \qquad \ts_0=\sigma^T\nu
\]
and then simulate the next return increment as
\[
R_1=\mu^TY_{t_0}\Delta_t+\ts_0\sqrt{\Delta_t}w_1
\]
with $w_1\sim\mathcal{N}(0,1)$. We see that for simulating $R_2$ we need an estimate of $\hy_1=\frac{\rho_1}{1^T\rho_1}$ to update the volatility process to its current value $\ts_1=\sigma^T\hy_1$.\\
The numerical method is thus as follows: Initialize $R_0, \rho_0,\ts_0$ as above.\\
If we assume that we have simulated the returns and the filter up to time $k-1$ we now need a way to simulate $R_k$ and calculate the filter. We discretize the return at time $t_k$ as
\[
R_k=\mu^TY_{t_{k-1}}\Delta_t+\ts_{k-1}\sqrt{\Delta_t}w_k,
\]
where $\ts_{k-1}=\sigma^T\hy_{k-1}=\frac{\rho_{k-1}}{1^T\rho_{k-1}}$.\\
Making use of the recursive nature of the filtering equation we can update the filter in the usual recursive way: To use the robust discretization of the filtering equation introduced in \cite{James1996} we first set
\[
\tilde{\phi}_k^i\colonequals\exp(\ts^{-2}_{k-1}\mu_iR_k-\frac{1}{2}\ts^{-2}_{k-1}\mu_i^2\Delta_t)
\]
and
\[
\tilde{\psi}_k=\diag(\tilde{\phi}_k^1,\dots,\tilde{\phi}_k^d).
\]
Then the next value of the discretized filter is given by 
\[
\rho_k=\tilde{\psi}_k(I+\Delta_tQ^T)\rho_{k-1}
\]
and the volatility is updated as
\[
\ts_{k}=\sigma^T\frac{\rho_{k}}{1^T\rho_{k}}.
\]
With this numerical scheme we can simulate the discretized returns.\\
\begin{figure} 
\centering
\input{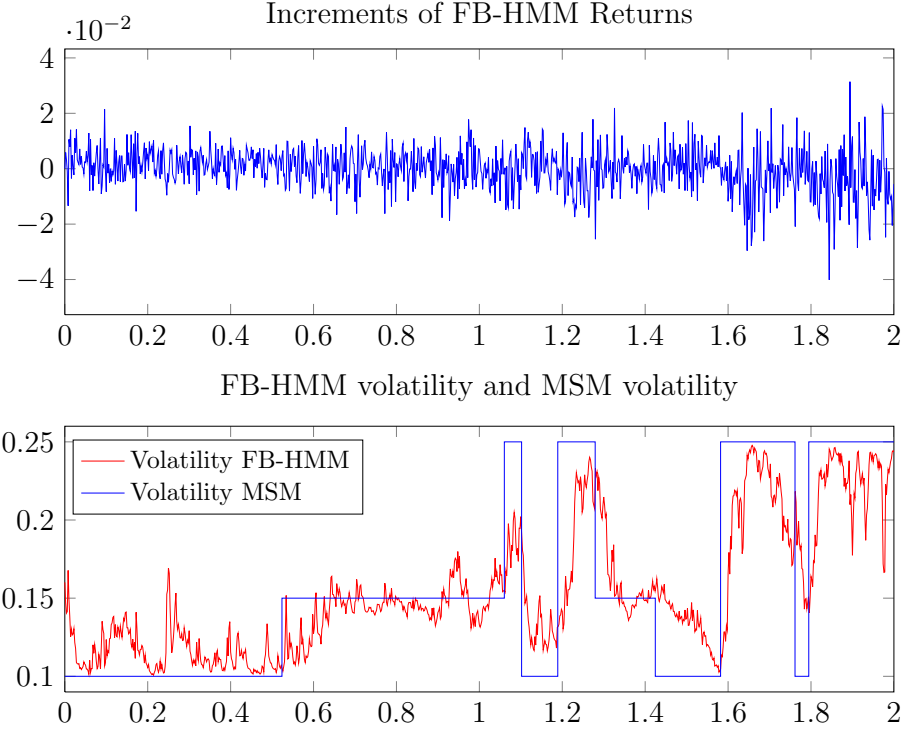} 
\caption{Return increments and volatility in the FB-HMM}
\label{fig_vola_sv} 
\end{figure}
\subsubsection*{Some Stylized Facts in the FB-HMM}
In Figure \ref{fig_vola_sv} we see a simulated path of the FB-HMM and the volatility processes both of the FB-HMM and the corresponding MSM.  The parameters were
\begin{align*}
Q=\begin{pmatrix}
-7 & 4 & 3\\
2 & -4 & 2\\
3 & 5 & -8
\end{pmatrix},\quad
\mu=\begin{pmatrix}
1\\
0\\
-2
\end{pmatrix}, \quad 
\sigma=\begin{pmatrix}
0.10\\
0.15\\
0.25
\end{pmatrix}.
\end{align*}
Both in the volatility process and the return increments it can be seen that the returns exhibit volatility clustering.\\
\begin{figure} 
\input{fig_asy_sv.tikz} 
\caption{Histogram of SV increments} 
\label{fig_asy_sv} 
\end{figure}
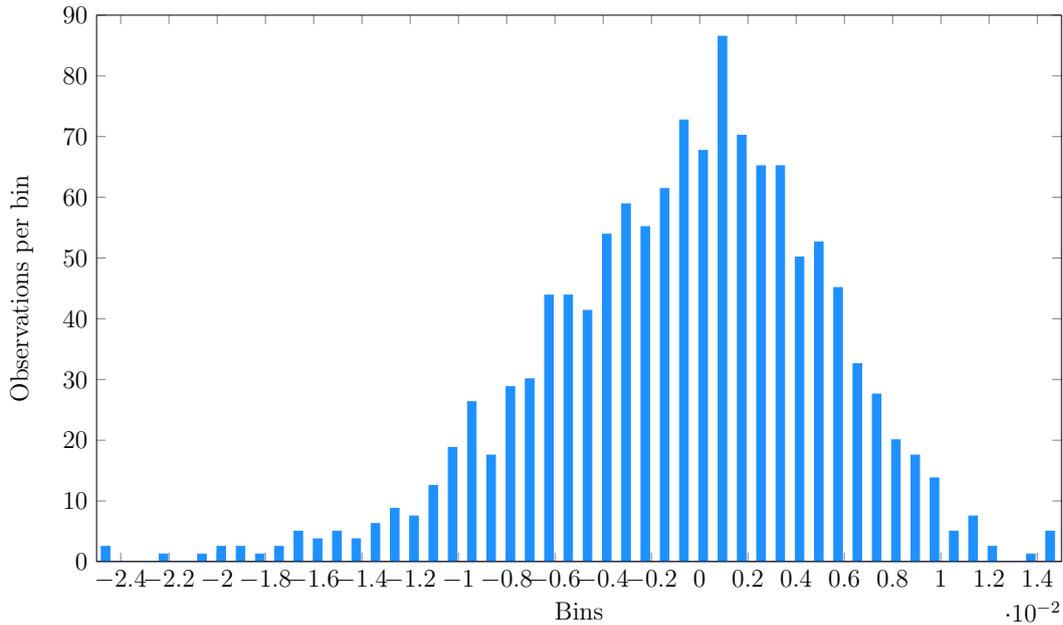
In Figure \ref{fig_asy_sv} we simulated increments of the FB-HMM over $T=4$ years with 250 discretization points each. We can clearly see that it is asymmetric. The left tail has more probability mass than the right tail. $Q$ was chosen as before and the parameters for drift and volatility were
\begin{align*}
\mu=\begin{pmatrix}
0.5\\
0\\
-1
\end{pmatrix}, \quad 
\sigma=\begin{pmatrix}
0.05\\
0.08\\
0.12
\end{pmatrix}.
\end{align*}
Choosing the largest volatility for the state with negative drift leads to this more pronounced left tail.

\bibliographystyle{alpha}
\bibliography{Ref_FB_vola}
\end{document}

%% file: fig_path_msm_hmm_2.tikz
%
%
\begin{tikzpicture}

\begin{axis}[%
width=4.521in,
height=0.962in,
at={(0.758in,5.703in)},
scale only axis,
separate axis lines,
every outer x axis line/.append style={black},
every x tick label/.append style={font=\color{black}},
xmin=0,
xmax=500,
every outer y axis line/.append style={black},
every y tick label/.append style={font=\color{black}},
ymin=-1.2849992137807,
ymax=1.26686997729171,
ylabel={Return},
axis background/.style={fill=white},
title={MSM Returns}
]
\addplot [color=red,solid,forget plot]
  table[row sep=crcr]{%
0	0\\
1	0.253064013201351\\
2	-0.105378136313494\\
3	0.106207399025857\\
4	-0.128097648746314\\
5	0.0412074858020942\\
6	0.103721152664301\\
7	0.205365190559947\\
8	-0.130911884942941\\
9	0.161530490780125\\
10	0.119762305559595\\
11	0.0285398751297934\\
12	0.0292510038748461\\
13	0.0502551134247259\\
14	0.208585341410469\\
15	0.232497538272029\\
16	-0.0673460835967828\\
17	-0.0826845912886505\\
18	-0.154914590301801\\
19	-0.160916409147211\\
20	0.0836227982513331\\
21	-0.0613746166606352\\
22	0.119855518528838\\
23	0.098178886662041\\
24	0.220488135174244\\
25	0.168095776185595\\
26	0.0943029973595465\\
27	0.183488309050985\\
28	-0.23755108290932\\
29	0.0194492146451781\\
30	-0.104552186344126\\
31	-0.251579985172408\\
32	0.244892587379683\\
33	0.22886856128132\\
34	0.0339768158619099\\
35	-0.196403062054973\\
36	0.0301766411265665\\
37	-0.0354553516634087\\
38	0.0641346775817173\\
39	0.156166381974221\\
40	-0.0734177215331045\\
41	-0.0692405330363589\\
42	0.0214305346546854\\
43	0.156483520252931\\
44	0.0641347890634461\\
45	0.00729298597511179\\
46	-0.151859115297677\\
47	0.0627185032855234\\
48	-0.130005472635914\\
49	-0.868526033467154\\
50	-0.0959904274464091\\
51	-0.50625458980729\\
52	0.00892366561155286\\
53	-0.00286329482769015\\
54	-0.0845571104207974\\
55	0.104275343204935\\
56	0.0487678490470383\\
57	0.269529726969017\\
58	-0.169044636660076\\
59	0.299217263014455\\
60	0.045932795561487\\
61	0.151253406336091\\
62	-0.871058517959773\\
63	-0.386183585097845\\
64	0.139534477299292\\
65	-0.0553389852664174\\
66	-0.0829483648457922\\
67	0.165663806416514\\
68	0.144619886093987\\
69	0.0485512547431325\\
70	-0.208084448551708\\
71	-0.766426753361188\\
72	-0.655086270110454\\
73	-0.272396325952011\\
74	-1.1849992137807\\
75	0.0801479978763344\\
76	0.0608690969414423\\
77	0.65046523934171\\
78	-0.266897634125233\\
79	0.607270245875349\\
80	-0.00312054635900015\\
81	0.161059655780269\\
82	-0.0941907846455149\\
83	-0.16576145354289\\
84	-0.0964375318263058\\
85	0.572079981014467\\
86	-0.502195784344388\\
87	-0.236807161397909\\
88	0.101338149095699\\
89	0.263592632202207\\
90	-0.587876934503921\\
91	-1.02158153360453\\
92	0.0476063719781381\\
93	0.153417112515508\\
94	0.238162598761471\\
95	0.0693545467648497\\
96	0.0358276791060407\\
97	0.0753171472066425\\
98	0.167546899918302\\
99	-0.0038023242323726\\
100	-0.16292796682855\\
101	0.0114221967872805\\
102	0.142275398271649\\
103	0.0959530610712818\\
104	-0.120831461864806\\
105	-0.125748544086911\\
106	0.0911283207839792\\
107	0.0350974504034409\\
108	-0.275878076437426\\
109	-0.162648933210638\\
110	0.190918760242356\\
111	0.0836364486703444\\
112	0.0101305669634949\\
113	-0.0885916466424119\\
114	0.145533053449681\\
115	-0.0174246625543967\\
116	-0.001185460138295\\
117	0.174992053581087\\
118	-0.113182359171336\\
119	-0.537558604330632\\
120	0.657210016756573\\
121	-0.369390870463406\\
122	-0.704141066459127\\
123	0.347269922798276\\
124	-0.238940618154753\\
125	-0.904179575331432\\
126	-0.0388872172547705\\
127	-0.257797334613039\\
128	-0.0377107270580038\\
129	-0.731454798199324\\
130	0.303073540012756\\
131	-1.10330938928753\\
132	0.116843701256941\\
133	0.467790581995764\\
134	0.229240986332457\\
135	0.0516749758656185\\
136	-0.0287619344402945\\
137	-0.0450524936572276\\
138	-0.174619786547209\\
139	-0.127954658608703\\
140	0.0784129947896809\\
141	0.0114871070597945\\
142	-0.038476882612593\\
143	0.0446405738157383\\
144	-0.15264094432823\\
145	0.284174553531232\\
146	0.305612034814454\\
147	0.240638356416214\\
148	0.231783225070745\\
149	0.137926401536828\\
150	0.222033143712774\\
151	0.0739642365501859\\
152	0.365056062953482\\
153	0.256009018701685\\
154	-0.0572602177981411\\
155	-0.189118129426607\\
156	-0.280673464040764\\
157	-0.432037759248867\\
158	-0.105327889874329\\
159	-0.295161690571177\\
160	-0.114643345774769\\
161	-0.196466445141887\\
162	-0.858525112927945\\
163	-0.0889783276316782\\
164	-0.392048685855607\\
165	0.0364896215386117\\
166	0.0514096582615037\\
167	0.462492907777645\\
168	-0.663195657381284\\
169	0.192456821896305\\
170	-0.803541145908106\\
171	0.00338140085913313\\
172	-0.0628439579505568\\
173	-0.173286769560826\\
174	0.182341427566858\\
175	0.172393403436951\\
176	-0.0223351159883079\\
177	-0.137196637475131\\
178	0.203622880388324\\
179	-0.0375745763924404\\
180	0.118407757223107\\
181	0.272825519374319\\
182	-0.107468483553746\\
183	0.193277545701699\\
184	-0.0453294832025182\\
185	-0.0834601061179063\\
186	-0.327631140144945\\
187	-0.27108267312648\\
188	-0.604143557199573\\
189	-0.626980367692618\\
190	0.238018254852004\\
191	0.14077296958231\\
192	0.294281367567122\\
193	0.184247667946493\\
194	0.0873591584018422\\
195	0.206550333861692\\
196	-0.119162400549806\\
197	0.0446687399610489\\
198	-0.0215591307497799\\
199	-0.299988209927665\\
200	0.123196951639618\\
201	0.157709177332728\\
202	0.0745894546168331\\
203	0.0462996319614286\\
204	0.0656582418968582\\
205	-0.172726750366618\\
206	-0.170035542017105\\
207	-0.246083711713996\\
208	0.133010542948195\\
209	0.145470903859324\\
210	-0.276463614038488\\
211	0.14974102351703\\
212	0.0189811560588634\\
213	0.0367387119266332\\
214	-0.0312324970738446\\
215	0.17890825688385\\
216	0.229578245982054\\
217	0.212177259046703\\
218	0.0822158331261083\\
219	0.0285947204929431\\
220	-0.0634459956422822\\
221	0.0876674117089837\\
222	0.0214923714864312\\
223	-0.157234710630089\\
224	0.0808215232866989\\
225	-0.0253983905400081\\
226	-0.17889729455394\\
227	0.166388884907777\\
228	0.0679862331313156\\
229	-0.0586750434847213\\
230	0.154105326073939\\
231	-0.00649993288385847\\
232	0.224944632712684\\
233	0.113985407197703\\
234	0.218121234232497\\
235	-0.178239655533557\\
236	0.0293235306527978\\
237	-0.129661396616707\\
238	-0.0100973389870715\\
239	-0.0698356988394019\\
240	-0.0554655131933665\\
241	0.0333123434386626\\
242	0.072598571120635\\
243	0.111802863119541\\
244	-0.134398911297272\\
245	0.0990781510872843\\
246	0.193563565031425\\
247	0.219073800381963\\
248	-0.0845359837726065\\
249	-0.0491587474247839\\
250	-0.0279465118205706\\
251	0.132366928460154\\
252	-0.0457934359766864\\
253	0.0165626711547473\\
254	0.346275076481422\\
255	0.172598510548003\\
256	-0.728263153820144\\
257	-0.194698628083591\\
258	-0.54670608882534\\
259	0.39160878634636\\
260	0.229109236221359\\
261	0.140637886828835\\
262	-0.197667486330097\\
263	0.0112192938586385\\
264	0.239015883821877\\
265	-0.0231103265029679\\
266	-0.00373870393245393\\
267	-0.0161065317330265\\
268	0.110650864935397\\
269	0.0522406338913494\\
270	0.0343369903967016\\
271	0.0745946569987494\\
272	0.051489758238326\\
273	0.0356598516768577\\
274	-0.0231236490591348\\
275	0.0361369026409548\\
276	0.151132537749309\\
277	-0.00395016408548235\\
278	0.0683936458555499\\
279	0.0379742307162148\\
280	0.26575885961946\\
281	0.0109944898625929\\
282	-0.154133529034936\\
283	-0.134639272550597\\
284	0.105684186713616\\
285	0.0622028928908555\\
286	0.0672874961249914\\
287	0.0669464692240035\\
288	-0.157236286645401\\
289	0.143152633350839\\
290	0.121263010719366\\
291	0.168710416178883\\
292	-0.186729316847484\\
293	0.0115536914802338\\
294	-0.0618326842019126\\
295	0.144764358228556\\
296	0.168877005264259\\
297	0.00868157835512291\\
298	-0.00258216454519604\\
299	-0.134487976782901\\
300	0.207618770461018\\
301	0.102818221551285\\
302	0.0934884352077965\\
303	-0.161741214513783\\
304	0.0467284413263502\\
305	0.150114941394347\\
306	0.100144461487858\\
307	-0.0331780307791654\\
308	-0.122732087535725\\
309	0.240779740997294\\
310	0.167186249905974\\
311	0.160056977481436\\
312	0.148117786744885\\
313	0.116030829423937\\
314	0.0301096151529164\\
315	0.145195205262148\\
316	0.0372954058610865\\
317	0.0368993095558469\\
318	-0.0300193280726963\\
319	0.119832980158576\\
320	0.00919826490048972\\
321	-0.390750979764199\\
322	-0.788455797960669\\
323	-0.230108928569363\\
324	-0.126211435550151\\
325	0.337375375673521\\
326	0.195560611867738\\
327	0.167348247595836\\
328	0.0688027312777174\\
329	-0.038564979159557\\
330	-0.189885497157836\\
331	-0.463101103823498\\
332	0.485668904089161\\
333	-0.282209527873471\\
334	0.103237167795787\\
335	-0.106464895270505\\
336	0.226758650936628\\
337	0.62283954650465\\
338	-0.197463498028337\\
339	-0.0488711942446546\\
340	0.100541147844396\\
341	-0.237141033338282\\
342	-0.16781923648018\\
343	0.0641311233722787\\
344	-0.230518500161185\\
345	0.160797663967803\\
346	0.0175719821877242\\
347	0.59738941011527\\
348	-0.55009762738753\\
349	-0.276210217915396\\
350	0.423364015689857\\
351	-0.390924961757485\\
352	-0.420425042360945\\
353	-0.808757754278483\\
354	-0.455739841134382\\
355	-0.40128746671184\\
356	0.00351948641248617\\
357	-0.344164956401726\\
358	-0.289375530181307\\
359	-0.0621578394616335\\
360	-0.381505673203484\\
361	-0.500370461759498\\
362	-0.130793765268356\\
363	-0.0877485266943006\\
364	-0.494692808926542\\
365	-0.526006934643275\\
366	0.319746253995617\\
367	-0.435213838589706\\
368	-0.48728567846431\\
369	-0.368665041164707\\
370	0.0547911743562925\\
371	-0.170586712626282\\
372	-0.586350965392575\\
373	-0.13278769317716\\
374	-0.645112819786066\\
375	1.16686997729171\\
376	0.0738626057784235\\
377	-0.289019763934772\\
378	-0.0800583939278429\\
379	-0.0415272506648392\\
380	0.841626169185471\\
381	0.0564275346853655\\
382	-0.738682342267829\\
383	0.267536418430357\\
384	-0.560178352476769\\
385	-0.0691381145580641\\
386	-0.11562716019237\\
387	0.595339048389854\\
388	-0.925142918368244\\
389	-0.0658121639204417\\
390	0.683628692067924\\
391	0.0724302910834475\\
392	-0.517924410267816\\
393	-0.352642439101594\\
394	-0.350657943845476\\
395	-0.245739226504746\\
396	0.658124252600382\\
397	-0.109591510986236\\
398	0.125081599307129\\
399	-0.879505800630882\\
400	-0.35460204877076\\
401	0.034090647997496\\
402	0.719215746071395\\
403	-0.416904180769519\\
404	-0.74327622436351\\
405	0.101590533927983\\
406	-0.297445467651374\\
407	-0.386277312614057\\
408	0.14763933536295\\
409	-0.174190966090237\\
410	0.103447325022318\\
411	0.172553217976287\\
412	0.14746915453115\\
413	-0.360275042886249\\
414	0.0915329176483327\\
415	0.138137604127087\\
416	0.0555461052463943\\
417	-0.958895341083003\\
418	0.321264598496761\\
419	-0.113748589215271\\
420	-0.551692200701873\\
421	0.0314839195649118\\
422	-0.173939472450438\\
423	-1.16338088950297\\
424	-0.433902586978998\\
425	0.366357244485325\\
426	-0.302799828144998\\
427	0.239645195763019\\
428	0.0700231547056773\\
429	0.0306270522786032\\
430	0.139192741555185\\
431	0.156483517102719\\
432	-0.084753112851442\\
433	-0.0984444614862552\\
434	0.186511191569564\\
435	0.132562248228331\\
436	0.298810838704264\\
437	0.19484188464002\\
438	-0.357521146236205\\
439	0.0248277037731035\\
440	-0.227164820018619\\
441	-0.0599434491516003\\
442	0.200770146886031\\
443	0.254005096682175\\
444	0.321931628659787\\
445	0.00146401635378164\\
446	0.0503904916646453\\
447	0.0263720194318114\\
448	-0.0329425553132003\\
449	-0.0197978192876223\\
450	-0.048721199945305\\
451	0.184414527019333\\
452	0.0598953602089992\\
453	-0.0432962961668855\\
454	-0.355970743528504\\
455	-0.0333831821874874\\
456	-0.124938417482175\\
457	0.209423355186555\\
458	0.119603505835216\\
459	0.103853220506359\\
460	0.319301977580358\\
461	0.0846720073119619\\
462	-0.0451779541445458\\
463	0.0551910522744816\\
464	0.104677013185884\\
465	0.164535786915128\\
466	0.0279703483877795\\
467	0.193330030627539\\
468	0.00393841522878974\\
469	0.301813890808114\\
470	-0.0723228598992527\\
471	0.119913803904692\\
472	0.0466469286889601\\
473	0.190066431565526\\
474	0.0783911183811599\\
475	0.0490762705849032\\
476	-0.117196732598925\\
477	0.207230480546916\\
478	0.283002719884318\\
479	0.0685964204406985\\
480	-0.180162302831356\\
481	0.301253838272366\\
482	-0.109270453769671\\
483	0.175872829484669\\
484	-0.00744204395559493\\
485	0.273040551955082\\
486	0.0117072810965702\\
487	-0.0944775703726478\\
488	-0.211913546269013\\
489	-0.61241278524089\\
490	0.017458750236076\\
491	-0.602394200776141\\
492	-0.237449352187801\\
493	0.0299715864094602\\
494	0.632357706441249\\
495	0.341052191526317\\
496	0.289231862180496\\
497	-0.494995642149123\\
498	-0.289152998850208\\
499	-0.837585842621695\\
500	-0.867357040323438\\
};
\end{axis}

\begin{axis}[%
width=4.521in,
height=0.962in,
at={(0.758in,4.105in)},
scale only axis,
separate axis lines,
every outer x axis line/.append style={black},
every x tick label/.append style={font=\color{black}},
xmin=0,
xmax=500,
every outer y axis line/.append style={black},
every y tick label/.append style={font=\color{black}},
ymin=-0.25,
ymax=0.15,
ylabel={Drift},
axis background/.style={fill=white},
title={Drift Process}
]
\addplot[const plot,color=myblue,solid,forget plot] plot table[row sep=crcr] {%
0	0\\
1	0\\
2	0\\
3	0\\
4	0\\
5	0\\
6	0.1\\
7	0.1\\
8	0\\
9	0.1\\
10	0.1\\
11	0.1\\
12	0.1\\
13	0.1\\
14	0\\
15	0.1\\
16	0.1\\
17	0\\
18	0\\
19	0\\
20	0.1\\
21	0\\
22	0.1\\
23	0\\
24	0\\
25	0.1\\
26	0.1\\
27	0\\
28	0\\
29	0\\
30	0\\
31	0\\
32	0\\
33	0\\
34	0\\
35	0\\
36	0\\
37	0\\
38	0.1\\
39	0.1\\
40	0\\
41	0\\
42	0.1\\
43	0.1\\
44	0.1\\
45	0\\
46	0\\
47	0\\
48	0.1\\
49	-0.2\\
50	-0.2\\
51	-0.2\\
52	-0.2\\
53	-0.2\\
54	-0.2\\
55	-0.2\\
56	0.1\\
57	0\\
58	0\\
59	-0.2\\
60	-0.2\\
61	-0.2\\
62	-0.2\\
63	-0.2\\
64	0\\
65	0.1\\
66	0\\
67	0.1\\
68	0.1\\
69	0.1\\
70	-0.2\\
71	-0.2\\
72	-0.2\\
73	-0.2\\
74	-0.2\\
75	0.1\\
76	-0.2\\
77	-0.2\\
78	-0.2\\
79	-0.2\\
80	0\\
81	0.1\\
82	0.1\\
83	0\\
84	-0.2\\
85	-0.2\\
86	-0.2\\
87	-0.2\\
88	-0.2\\
89	-0.2\\
90	-0.2\\
91	-0.2\\
92	0.1\\
93	0\\
94	0.1\\
95	0\\
96	0\\
97	0\\
98	0.1\\
99	0\\
100	0\\
101	0\\
102	0\\
103	0\\
104	0\\
105	0\\
106	0.1\\
107	0\\
108	0\\
109	0\\
110	0.1\\
111	0\\
112	0\\
113	0.1\\
114	0.1\\
115	0.1\\
116	0\\
117	0\\
118	0.1\\
119	-0.2\\
120	-0.2\\
121	-0.2\\
122	-0.2\\
123	-0.2\\
124	-0.2\\
125	-0.2\\
126	-0.2\\
127	-0.2\\
128	-0.2\\
129	-0.2\\
130	-0.2\\
131	-0.2\\
132	-0.2\\
133	-0.2\\
134	0\\
135	0\\
136	0\\
137	0.1\\
138	0\\
139	0.1\\
140	0\\
141	0\\
142	-0.2\\
143	-0.2\\
144	-0.2\\
145	0\\
146	0.1\\
147	0.1\\
148	0\\
149	0\\
150	0\\
151	0.1\\
152	0.1\\
153	0.1\\
154	0\\
155	-0.2\\
156	-0.2\\
157	-0.2\\
158	-0.2\\
159	-0.2\\
160	-0.2\\
161	-0.2\\
162	-0.2\\
163	-0.2\\
164	-0.2\\
165	-0.2\\
166	-0.2\\
167	-0.2\\
168	-0.2\\
169	-0.2\\
170	-0.2\\
171	-0.2\\
172	-0.2\\
173	0\\
174	0\\
175	0.1\\
176	0.1\\
177	0\\
178	0\\
179	0\\
180	0\\
181	0\\
182	0\\
183	0\\
184	0\\
185	0.1\\
186	-0.2\\
187	-0.2\\
188	-0.2\\
189	-0.2\\
190	-0.2\\
191	0.1\\
192	0.1\\
193	0\\
194	0\\
195	0\\
196	0\\
197	0.1\\
198	0\\
199	0\\
200	0\\
201	0.1\\
202	0.1\\
203	0.1\\
204	0\\
205	0\\
206	0\\
207	0\\
208	0\\
209	0.1\\
210	0\\
211	0\\
212	0\\
213	0\\
214	0\\
215	0.1\\
216	0.1\\
217	0.1\\
218	0.1\\
219	0\\
220	0\\
221	0\\
222	0\\
223	0\\
224	0\\
225	0\\
226	0\\
227	0\\
228	0.1\\
229	0\\
230	0\\
231	0\\
232	0\\
233	0\\
234	0.1\\
235	0\\
236	0\\
237	0\\
238	0\\
239	0\\
240	0\\
241	0\\
242	0.1\\
243	0\\
244	0\\
245	0\\
246	0\\
247	0.1\\
248	0\\
249	0\\
250	0.1\\
251	0\\
252	0\\
253	0\\
254	-0.2\\
255	-0.2\\
256	-0.2\\
257	-0.2\\
258	-0.2\\
259	-0.2\\
260	-0.2\\
261	0\\
262	0\\
263	0\\
264	0.1\\
265	0.1\\
266	0\\
267	0\\
268	0.1\\
269	0.1\\
270	0.1\\
271	0\\
272	0.1\\
273	0\\
274	0\\
275	0.1\\
276	0.1\\
277	0\\
278	0\\
279	0.1\\
280	0\\
281	0\\
282	0.1\\
283	0\\
284	0\\
285	0.1\\
286	0.1\\
287	0.1\\
288	0\\
289	0.1\\
290	0\\
291	0\\
292	0\\
293	0\\
294	0\\
295	0\\
296	0.1\\
297	0\\
298	0\\
299	0\\
300	0.1\\
301	0\\
302	0\\
303	0\\
304	0.1\\
305	0.1\\
306	0.1\\
307	0\\
308	0\\
309	0\\
310	0.1\\
311	0.1\\
312	0.1\\
313	0\\
314	0\\
315	0\\
316	0\\
317	0\\
318	0\\
319	0.1\\
320	0\\
321	-0.2\\
322	-0.2\\
323	0\\
324	0\\
325	0.1\\
326	0.1\\
327	0.1\\
328	0.1\\
329	0\\
330	0\\
331	-0.2\\
332	-0.2\\
333	-0.2\\
334	-0.2\\
335	-0.2\\
336	-0.2\\
337	-0.2\\
338	-0.2\\
339	-0.2\\
340	-0.2\\
341	-0.2\\
342	0\\
343	0.1\\
344	0\\
345	0.1\\
346	-0.2\\
347	-0.2\\
348	-0.2\\
349	-0.2\\
350	-0.2\\
351	-0.2\\
352	-0.2\\
353	-0.2\\
354	-0.2\\
355	-0.2\\
356	-0.2\\
357	-0.2\\
358	-0.2\\
359	-0.2\\
360	-0.2\\
361	-0.2\\
362	-0.2\\
363	-0.2\\
364	-0.2\\
365	-0.2\\
366	-0.2\\
367	-0.2\\
368	-0.2\\
369	-0.2\\
370	-0.2\\
371	-0.2\\
372	-0.2\\
373	-0.2\\
374	-0.2\\
375	-0.2\\
376	-0.2\\
377	-0.2\\
378	-0.2\\
379	-0.2\\
380	-0.2\\
381	-0.2\\
382	-0.2\\
383	-0.2\\
384	-0.2\\
385	-0.2\\
386	-0.2\\
387	-0.2\\
388	-0.2\\
389	-0.2\\
390	-0.2\\
391	-0.2\\
392	-0.2\\
393	-0.2\\
394	-0.2\\
395	-0.2\\
396	-0.2\\
397	-0.2\\
398	-0.2\\
399	-0.2\\
400	-0.2\\
401	-0.2\\
402	-0.2\\
403	-0.2\\
404	-0.2\\
405	-0.2\\
406	-0.2\\
407	-0.2\\
408	-0.2\\
409	-0.2\\
410	-0.2\\
411	0.1\\
412	0.1\\
413	0\\
414	0.1\\
415	0\\
416	0\\
417	-0.2\\
418	-0.2\\
419	-0.2\\
420	-0.2\\
421	-0.2\\
422	-0.2\\
423	-0.2\\
424	-0.2\\
425	-0.2\\
426	-0.2\\
427	0\\
428	0.1\\
429	0\\
430	0.1\\
431	0.1\\
432	0\\
433	0\\
434	0.1\\
435	0.1\\
436	0.1\\
437	0.1\\
438	0\\
439	0\\
440	0\\
441	0\\
442	0.1\\
443	0.1\\
444	0\\
445	0.1\\
446	0.1\\
447	0\\
448	0\\
449	0\\
450	0\\
451	0\\
452	0\\
453	0.1\\
454	-0.2\\
455	-0.2\\
456	-0.2\\
457	0\\
458	0.1\\
459	0.1\\
460	0.1\\
461	0\\
462	0.1\\
463	0\\
464	0\\
465	0\\
466	0.1\\
467	0.1\\
468	0\\
469	0\\
470	0\\
471	0.1\\
472	0\\
473	0.1\\
474	0.1\\
475	0\\
476	0\\
477	0.1\\
478	0.1\\
479	0\\
480	0\\
481	0.1\\
482	0\\
483	0.1\\
484	0.1\\
485	0\\
486	0\\
487	0\\
488	0\\
489	-0.2\\
490	-0.2\\
491	-0.2\\
492	-0.2\\
493	-0.2\\
494	-0.2\\
495	-0.2\\
496	-0.2\\
497	-0.2\\
498	-0.2\\
499	-0.2\\
500	-0.2\\
};
\end{axis}

\begin{axis}[%
width=4.521in,
height=0.962in,
at={(0.758in,2.508in)},
scale only axis,
every outer x axis line/.append style={black},
every x tick label/.append style={font=\color{black}},
xmin=0,
xmax=500,
every outer y axis line/.append style={black},
every y tick label/.append style={font=\color{black}},
ymin=0.08,
ymax=0.42,
ylabel={Volatility},
axis background/.style={fill=white},
title={MSM Volatility},
legend style={legend cell align=left,align=left,draw=black}
]
\addplot[const plot,color=red,solid] plot table[row sep=crcr] {%
0	0.15\\
1	0.15\\
2	0.15\\
3	0.15\\
4	0.15\\
5	0.15\\
6	0.1\\
7	0.1\\
8	0.15\\
9	0.1\\
10	0.1\\
11	0.1\\
12	0.1\\
13	0.1\\
14	0.15\\
15	0.1\\
16	0.1\\
17	0.15\\
18	0.15\\
19	0.15\\
20	0.1\\
21	0.15\\
22	0.1\\
23	0.15\\
24	0.15\\
25	0.1\\
26	0.1\\
27	0.15\\
28	0.15\\
29	0.15\\
30	0.15\\
31	0.15\\
32	0.15\\
33	0.15\\
34	0.15\\
35	0.15\\
36	0.15\\
37	0.15\\
38	0.1\\
39	0.1\\
40	0.15\\
41	0.15\\
42	0.1\\
43	0.1\\
44	0.1\\
45	0.15\\
46	0.15\\
47	0.15\\
48	0.1\\
49	0.4\\
50	0.4\\
51	0.4\\
52	0.4\\
53	0.4\\
54	0.4\\
55	0.4\\
56	0.1\\
57	0.15\\
58	0.15\\
59	0.4\\
60	0.4\\
61	0.4\\
62	0.4\\
63	0.4\\
64	0.15\\
65	0.1\\
66	0.15\\
67	0.1\\
68	0.1\\
69	0.1\\
70	0.4\\
71	0.4\\
72	0.4\\
73	0.4\\
74	0.4\\
75	0.1\\
76	0.4\\
77	0.4\\
78	0.4\\
79	0.4\\
80	0.15\\
81	0.1\\
82	0.1\\
83	0.15\\
84	0.4\\
85	0.4\\
86	0.4\\
87	0.4\\
88	0.4\\
89	0.4\\
90	0.4\\
91	0.4\\
92	0.1\\
93	0.15\\
94	0.1\\
95	0.15\\
96	0.15\\
97	0.15\\
98	0.1\\
99	0.15\\
100	0.15\\
101	0.15\\
102	0.15\\
103	0.15\\
104	0.15\\
105	0.15\\
106	0.1\\
107	0.15\\
108	0.15\\
109	0.15\\
110	0.1\\
111	0.15\\
112	0.15\\
113	0.1\\
114	0.1\\
115	0.1\\
116	0.15\\
117	0.15\\
118	0.1\\
119	0.4\\
120	0.4\\
121	0.4\\
122	0.4\\
123	0.4\\
124	0.4\\
125	0.4\\
126	0.4\\
127	0.4\\
128	0.4\\
129	0.4\\
130	0.4\\
131	0.4\\
132	0.4\\
133	0.4\\
134	0.15\\
135	0.15\\
136	0.15\\
137	0.1\\
138	0.15\\
139	0.1\\
140	0.15\\
141	0.15\\
142	0.4\\
143	0.4\\
144	0.4\\
145	0.15\\
146	0.1\\
147	0.1\\
148	0.15\\
149	0.15\\
150	0.15\\
151	0.1\\
152	0.1\\
153	0.1\\
154	0.15\\
155	0.4\\
156	0.4\\
157	0.4\\
158	0.4\\
159	0.4\\
160	0.4\\
161	0.4\\
162	0.4\\
163	0.4\\
164	0.4\\
165	0.4\\
166	0.4\\
167	0.4\\
168	0.4\\
169	0.4\\
170	0.4\\
171	0.4\\
172	0.4\\
173	0.15\\
174	0.15\\
175	0.1\\
176	0.1\\
177	0.15\\
178	0.15\\
179	0.15\\
180	0.15\\
181	0.15\\
182	0.15\\
183	0.15\\
184	0.15\\
185	0.1\\
186	0.4\\
187	0.4\\
188	0.4\\
189	0.4\\
190	0.4\\
191	0.1\\
192	0.1\\
193	0.15\\
194	0.15\\
195	0.15\\
196	0.15\\
197	0.1\\
198	0.15\\
199	0.15\\
200	0.15\\
201	0.1\\
202	0.1\\
203	0.1\\
204	0.15\\
205	0.15\\
206	0.15\\
207	0.15\\
208	0.15\\
209	0.1\\
210	0.15\\
211	0.15\\
212	0.15\\
213	0.15\\
214	0.15\\
215	0.1\\
216	0.1\\
217	0.1\\
218	0.1\\
219	0.15\\
220	0.15\\
221	0.15\\
222	0.15\\
223	0.15\\
224	0.15\\
225	0.15\\
226	0.15\\
227	0.15\\
228	0.1\\
229	0.15\\
230	0.15\\
231	0.15\\
232	0.15\\
233	0.15\\
234	0.1\\
235	0.15\\
236	0.15\\
237	0.15\\
238	0.15\\
239	0.15\\
240	0.15\\
241	0.15\\
242	0.1\\
243	0.15\\
244	0.15\\
245	0.15\\
246	0.15\\
247	0.1\\
248	0.15\\
249	0.15\\
250	0.1\\
251	0.15\\
252	0.15\\
253	0.15\\
254	0.4\\
255	0.4\\
256	0.4\\
257	0.4\\
258	0.4\\
259	0.4\\
260	0.4\\
261	0.15\\
262	0.15\\
263	0.15\\
264	0.1\\
265	0.1\\
266	0.15\\
267	0.15\\
268	0.1\\
269	0.1\\
270	0.1\\
271	0.15\\
272	0.1\\
273	0.15\\
274	0.15\\
275	0.1\\
276	0.1\\
277	0.15\\
278	0.15\\
279	0.1\\
280	0.15\\
281	0.15\\
282	0.1\\
283	0.15\\
284	0.15\\
285	0.1\\
286	0.1\\
287	0.1\\
288	0.15\\
289	0.1\\
290	0.15\\
291	0.15\\
292	0.15\\
293	0.15\\
294	0.15\\
295	0.15\\
296	0.1\\
297	0.15\\
298	0.15\\
299	0.15\\
300	0.1\\
301	0.15\\
302	0.15\\
303	0.15\\
304	0.1\\
305	0.1\\
306	0.1\\
307	0.15\\
308	0.15\\
309	0.15\\
310	0.1\\
311	0.1\\
312	0.1\\
313	0.15\\
314	0.15\\
315	0.15\\
316	0.15\\
317	0.15\\
318	0.15\\
319	0.1\\
320	0.15\\
321	0.4\\
322	0.4\\
323	0.15\\
324	0.15\\
325	0.1\\
326	0.1\\
327	0.1\\
328	0.1\\
329	0.15\\
330	0.15\\
331	0.4\\
332	0.4\\
333	0.4\\
334	0.4\\
335	0.4\\
336	0.4\\
337	0.4\\
338	0.4\\
339	0.4\\
340	0.4\\
341	0.4\\
342	0.15\\
343	0.1\\
344	0.15\\
345	0.1\\
346	0.4\\
347	0.4\\
348	0.4\\
349	0.4\\
350	0.4\\
351	0.4\\
352	0.4\\
353	0.4\\
354	0.4\\
355	0.4\\
356	0.4\\
357	0.4\\
358	0.4\\
359	0.4\\
360	0.4\\
361	0.4\\
362	0.4\\
363	0.4\\
364	0.4\\
365	0.4\\
366	0.4\\
367	0.4\\
368	0.4\\
369	0.4\\
370	0.4\\
371	0.4\\
372	0.4\\
373	0.4\\
374	0.4\\
375	0.4\\
376	0.4\\
377	0.4\\
378	0.4\\
379	0.4\\
380	0.4\\
381	0.4\\
382	0.4\\
383	0.4\\
384	0.4\\
385	0.4\\
386	0.4\\
387	0.4\\
388	0.4\\
389	0.4\\
390	0.4\\
391	0.4\\
392	0.4\\
393	0.4\\
394	0.4\\
395	0.4\\
396	0.4\\
397	0.4\\
398	0.4\\
399	0.4\\
400	0.4\\
401	0.4\\
402	0.4\\
403	0.4\\
404	0.4\\
405	0.4\\
406	0.4\\
407	0.4\\
408	0.4\\
409	0.4\\
410	0.4\\
411	0.1\\
412	0.1\\
413	0.15\\
414	0.1\\
415	0.15\\
416	0.15\\
417	0.4\\
418	0.4\\
419	0.4\\
420	0.4\\
421	0.4\\
422	0.4\\
423	0.4\\
424	0.4\\
425	0.4\\
426	0.4\\
427	0.15\\
428	0.1\\
429	0.15\\
430	0.1\\
431	0.1\\
432	0.15\\
433	0.15\\
434	0.1\\
435	0.1\\
436	0.1\\
437	0.1\\
438	0.15\\
439	0.15\\
440	0.15\\
441	0.15\\
442	0.1\\
443	0.1\\
444	0.15\\
445	0.1\\
446	0.1\\
447	0.15\\
448	0.15\\
449	0.15\\
450	0.15\\
451	0.15\\
452	0.15\\
453	0.1\\
454	0.4\\
455	0.4\\
456	0.4\\
457	0.15\\
458	0.1\\
459	0.1\\
460	0.1\\
461	0.15\\
462	0.1\\
463	0.15\\
464	0.15\\
465	0.15\\
466	0.1\\
467	0.1\\
468	0.15\\
469	0.15\\
470	0.15\\
471	0.1\\
472	0.15\\
473	0.1\\
474	0.1\\
475	0.15\\
476	0.15\\
477	0.1\\
478	0.1\\
479	0.15\\
480	0.15\\
481	0.1\\
482	0.15\\
483	0.1\\
484	0.1\\
485	0.15\\
486	0.15\\
487	0.15\\
488	0.15\\
489	0.4\\
490	0.4\\
491	0.4\\
492	0.4\\
493	0.4\\
494	0.4\\
495	0.4\\
496	0.4\\
497	0.4\\
498	0.4\\
499	0.4\\
500	0.4\\
};
\end{axis}

\begin{axis}[%
width=4.521in,
height=0.962in,
at={(0.758in,0.91in)},
scale only axis,
separate axis lines,
every outer x axis line/.append style={black},
every x tick label/.append style={font=\color{black}},
xmin=0,
xmax=500,
xlabel={Time},
every outer y axis line/.append style={black},
every y tick label/.append style={font=\color{black}},
ymin=-1.2849992137807,
ymax=1.26686997729171,
ylabel={Return},
axis background/.style={fill=white},
title={HMM Returns}
]
\addplot [color=myblue,solid,forget plot]
  table[row sep=crcr]{%
0	0\\
1	0.37546098690658\\
2	-0.156345339497798\\
3	0.157575683522024\\
4	-0.190053374283747\\
5	0.0611379037717345\\
6	0.108281388772513\\
7	0.334489198599097\\
8	-0.194228744327108\\
9	0.236935503991063\\
10	0.143980817274785\\
11	-0.0590338073091854\\
12	-0.0574511972589209\\
13	-0.0107067573783062\\
14	0.309469754903115\\
15	0.394871972428928\\
16	-0.272427068396762\\
17	-0.122675831519763\\
18	-0.229840601297444\\
19	-0.238745260630176\\
20	0.0635526980691434\\
21	-0.0910590717775437\\
22	0.144188261823982\\
23	0.145664099818845\\
24	0.327129455454597\\
25	0.251546482295393\\
26	0.0873213764766377\\
27	0.272234288591984\\
28	-0.352445070721671\\
29	0.028856024342846\\
30	-0.155119910458278\\
31	-0.373259193687167\\
32	0.363337368203843\\
33	0.339563159548102\\
34	0.0504100470630951\\
35	-0.291395392722084\\
36	0.0447718793185007\\
37	-0.0526036916836195\\
38	0.02018207657892\\
39	0.224997732432824\\
40	-0.108926946326894\\
41	-0.102729418295774\\
42	-0.0748555748371217\\
43	0.225703520562896\\
44	0.0201823246804143\\
45	0.0108203125251659\\
46	-0.225307314853417\\
47	0.0930529427830969\\
48	-0.411874924395613\\
49	-0.571949533326088\\
50	-0.142131928995919\\
51	-0.370391646780037\\
52	-0.0837606076131801\\
53	-0.0903185488379551\\
54	-0.13577074525863\\
55	-0.0307095516972542\\
56	-0.0140166496698265\\
57	0.399890509947496\\
58	-0.250804787724427\\
59	0.0777507811379443\\
60	-0.0631697436459374\\
61	-0.00457224696496902\\
62	-0.573358538178599\\
63	-0.303587435826497\\
64	0.207021740829668\\
65	-0.245705388779184\\
66	-0.123067181830032\\
67	0.246134157417145\\
68	0.199301119052305\\
69	-0.0144986781696953\\
70	-0.204497965248132\\
71	-0.515144296600465\\
72	-0.453197508125179\\
73	-0.240279328409575\\
74	-0.748026523353479\\
75	0.055819563901254\\
76	-0.0548595955742466\\
77	0.273175513065118\\
78	-0.237220007221637\\
79	0.249143004445354\\
80	-0.00462983021890873\\
81	0.235887665314911\\
82	-0.332169687397372\\
83	-0.245933659831608\\
84	-0.14238068559944\\
85	0.229564107084029\\
86	-0.368133438838667\\
87	-0.220478494209131\\
88	-0.0323437258707755\\
89	0.0579302144850514\\
90	-0.415804078755858\\
91	-0.657105412078989\\
92	-0.0166015054996338\\
93	0.227618853209284\\
94	0.407479509008372\\
95	0.102898575919091\\
96	0.0531560990658251\\
97	0.111745048469986\\
98	0.250324963543672\\
99	-0.00564135686763778\\
100	-0.24172972856262\\
101	0.0169466579785143\\
102	0.211088335997806\\
103	0.142361731131902\\
104	-0.179272822505301\\
105	-0.186568101357704\\
106	0.0802561648819928\\
107	0.0520726878534712\\
108	-0.409309302949645\\
109	-0.241315737508593\\
110	0.302338809558968\\
111	0.124088064367112\\
112	0.015030318305316\\
113	-0.319708860664976\\
114	0.201333364049781\\
115	-0.161327435292628\\
116	-0.00175881994374488\\
117	0.259628733090894\\
118	-0.374435250312679\\
119	-0.387808341134935\\
120	0.276928122067995\\
121	-0.294244430380375\\
122	-0.480490250211328\\
123	0.104485961949041\\
124	-0.221665490983159\\
125	-0.591786185294694\\
126	-0.110361270384394\\
127	-0.232156850385196\\
128	-0.109706703534723\\
129	-0.495686860762859\\
130	0.0798963078012145\\
131	-0.702576547471246\\
132	-0.0237168622908688\\
133	0.171540348316271\\
134	0.340115711748156\\
135	0.0766681014476824\\
136	-0.0426729354114174\\
137	-0.222812902550889\\
138	-0.259076415334748\\
139	-0.407310857884074\\
140	0.116338234099723\\
141	0.0170429627619174\\
142	-0.110132971453575\\
143	-0.0638887003525181\\
144	-0.173650721476736\\
145	0.421618455239147\\
146	0.557587567675304\\
147	0.412989283396868\\
148	0.343887530006922\\
149	0.204635902933725\\
150	0.329421723024835\\
151	0.0420576636950215\\
152	0.689879669514122\\
153	0.447196541620416\\
154	-0.0849547022233858\\
155	-0.19394562593098\\
156	-0.244884500826602\\
157	-0.32909943958209\\
158	-0.147327036768315\\
159	-0.252945352352101\\
160	-0.152509900712923\\
161	-0.198034027076491\\
162	-0.566385295673146\\
163	-0.138230589148017\\
164	-0.30685061688535\\
165	-0.0684236664478803\\
166	-0.0601225675848986\\
167	0.168592867807661\\
168	-0.457709348592038\\
169	0.0183522023785813\\
170	-0.535793725787108\\
171	-0.0868441715808254\\
172	-0.123690143271511\\
173	-0.257098671178481\\
174	0.270532706259326\\
175	0.261110809609686\\
176	-0.172255601268097\\
177	-0.203553181090554\\
178	0.302107149334311\\
179	-0.0557479009221175\\
180	0.175676868559774\\
181	0.404780345738369\\
182	-0.159446704357518\\
183	0.286758188720821\\
184	-0.0672535469736707\\
185	-0.308288667536909\\
186	-0.271010462776722\\
187	-0.239548448038507\\
188	-0.424854381088978\\
189	-0.437560155554472\\
190	0.0437013329691295\\
191	0.190739844070435\\
192	0.53237127880134\\
193	0.273360919110156\\
194	0.129611300373975\\
195	0.306450495337282\\
196	-0.17679650277651\\
197	-0.023139176753352\\
198	-0.0319864227464055\\
199	-0.445080546755424\\
200	0.182782405373813\\
201	0.228431208377739\\
202	0.043449080372756\\
203	-0.0195096425956443\\
204	0.0974145157554693\\
205	-0.256267793027596\\
206	-0.252274954495967\\
207	-0.365104591889393\\
208	0.197342439537517\\
209	0.20119505074575\\
210	-0.410178041743378\\
211	0.222164786525267\\
212	0.0281615844794901\\
213	0.0545077621395146\\
214	-0.0463384106912596\\
215	0.275609552084647\\
216	0.388375116058101\\
217	0.349649390231389\\
218	0.0604215109767311\\
219	0.0424248467444319\\
220	-0.0941322941882226\\
221	0.130068643515943\\
222	0.0318873746890189\\
223	-0.233282871326994\\
224	0.119911671804449\\
225	-0.0376825794286395\\
226	-0.265422783423165\\
227	0.246864554732453\\
228	0.0287536756942023\\
229	-0.087053822686482\\
230	0.228639928227347\\
231	-0.00964369127212988\\
232	0.333741383175028\\
233	0.169115604142997\\
234	0.362877648733107\\
235	-0.264447070628219\\
236	0.0435061533214713\\
237	-0.192373444653546\\
238	-0.014981019281472\\
239	-0.103612442068917\\
240	-0.0822919705548641\\
241	0.0494241958207609\\
242	0.0390183886704327\\
243	0.165877450510691\\
244	-0.199402306303796\\
245	0.146998302593553\\
246	0.287182544196951\\
247	0.364997575359858\\
248	-0.125422668734521\\
249	-0.072934873630235\\
250	-0.184743707679113\\
251	0.196387534382058\\
252	-0.0679418952072406\\
253	0.0245733748505075\\
254	0.10393245676785\\
255	0.00730358307450155\\
256	-0.493911117444051\\
257	-0.197050462193566\\
258	-0.392897750400373\\
259	0.129154888481921\\
260	0.0387445995643343\\
261	0.208658825556507\\
262	-0.293271368607399\\
263	0.016645618992882\\
264	0.409378486544766\\
265	-0.17398082466837\\
266	-0.00554696596514406\\
267	-0.0238966189764511\\
268	0.123703395493482\\
269	-0.00628800104572234\\
270	-0.0461323841171445\\
271	0.1106731185537\\
272	-0.00795906744999994\\
273	0.0529071002003053\\
274	-0.0343076361857752\\
275	-0.0421266970637576\\
276	0.213794961461697\\
277	-0.00586070096342806\\
278	0.101472925550391\\
279	-0.0380377414452866\\
280	0.394295824402729\\
281	0.0163120862667228\\
282	-0.465571677362063\\
283	-0.199758920712323\\
284	0.156799414274449\\
285	0.0158829086884726\\
286	0.0271986433369906\\
287	0.026439691312243\\
288	-0.233285209598078\\
289	0.196035762457259\\
290	0.179913094335269\\
291	0.250308918121611\\
292	-0.277042842642999\\
293	0.0171417514118502\\
294	-0.0917386883257135\\
295	0.214781106652824\\
296	0.253285099950851\\
297	0.0128805116772085\\
298	-0.00383105458666341\\
299	-0.199534449213847\\
300	0.33950451857501\\
301	0.152547296027069\\
302	0.13870506400111\\
303	-0.239968991468162\\
304	-0.0185553315580247\\
305	0.211530310750164\\
306	0.10032149762494\\
307	-0.049224921482814\\
308	-0.182092705036664\\
309	0.357235302002521\\
310	0.24952234047702\\
311	0.23365621459104\\
312	0.207085662657735\\
313	0.172150315550547\\
314	0.0446724355536734\\
315	0.215420337219004\\
316	0.055333706735076\\
317	0.054746034439067\\
318	-0.0445384802124317\\
319	0.144138102901929\\
320	0.0136470989046481\\
321	-0.306128608839395\\
322	-0.527400652296745\\
323	-0.34140344304082\\
324	-0.187254874966564\\
325	0.628276571351857\\
326	0.312669204842908\\
327	0.249882864747596\\
328	0.0305707843141358\\
329	-0.0572173220210421\\
330	-0.281725541534829\\
331	-0.346382231784152\\
332	0.181487355951567\\
333	-0.245739124576662\\
334	-0.0312871639959715\\
335	-0.147959635358835\\
336	0.0374367984377806\\
337	0.25780533592293\\
338	-0.198588759932433\\
339	-0.115916081111609\\
340	-0.0327871554885349\\
341	-0.220664251391642\\
342	-0.248986710333339\\
343	0.0201741667206595\\
344	-0.342011108082282\\
345	0.235304605104817\\
346	-0.0789489216749671\\
347	0.243645578667074\\
348	-0.394784709355317\\
349	-0.242401273202929\\
350	0.146822626376464\\
351	-0.306225407644483\\
352	-0.322638442686114\\
353	-0.538696103483371\\
354	-0.342286627297806\\
355	-0.311990820940166\\
356	-0.0867673445695236\\
357	-0.280209424272529\\
358	-0.249726091546953\\
359	-0.123308405778899\\
360	-0.300984774061743\\
361	-0.367117879459329\\
362	-0.161495550774306\\
363	-0.137546361665702\\
364	-0.363958989280208\\
365	-0.381381309225548\\
366	0.0891725481789342\\
367	-0.330866522940841\\
368	-0.359837865224016\\
369	-0.293840598883305\\
370	-0.0582411848556902\\
371	-0.183635254328838\\
372	-0.414955071431653\\
373	-0.162604917527489\\
374	-0.447648554145679\\
375	0.56048893344416\\
376	-0.0476303639419065\\
377	-0.249528152973513\\
378	-0.133267782896128\\
379	-0.111830112502251\\
380	0.37953220687525\\
381	-0.0573307588882893\\
382	-0.499708067879405\\
383	0.0601244288815956\\
384	-0.400393348069183\\
385	-0.127192039227158\\
386	-0.153057268048206\\
387	0.242504813687492\\
388	-0.603449613896058\\
389	-0.125341571592991\\
390	0.291626747792693\\
391	-0.0484272645197485\\
392	-0.376884414536261\\
393	-0.284926062931524\\
394	-0.28382194424736\\
395	-0.225448050040631\\
396	0.27743677779482\\
397	-0.149699198514401\\
398	-0.0191335219541216\\
399	-0.578058374370614\\
400	-0.286016335958241\\
401	-0.0697583894719814\\
402	0.311426407740703\\
403	-0.320679531947747\\
404	-0.502263977770875\\
405	-0.0322033058783035\\
406	-0.254215983227603\\
407	-0.303639583243605\\
408	-0.00658301684463308\\
409	-0.185640562015892\\
410	-0.0311702382841513\\
411	0.261466475300167\\
412	0.205642138025205\\
413	-0.534525717223389\\
414	0.0811565912369756\\
415	0.204949255796397\\
416	0.0824115417707942\\
417	-0.622228535357455\\
418	0.0900173133793252\\
419	-0.152012082725163\\
420	-0.395671886174817\\
421	-0.0712087016146201\\
422	-0.185500637858454\\
423	-0.735998681169546\\
424	-0.330136978539786\\
425	0.115105623770021\\
426	-0.25719500242381\\
427	0.35555202247193\\
428	0.0332868246881249\\
429	0.0454401363872087\\
430	0.1872230620885\\
431	0.225703513552129\\
432	-0.125744814492009\\
433	-0.146058122597254\\
434	0.292529808689127\\
435	0.172466964194423\\
436	0.542451572410469\\
437	0.311069684443966\\
438	-0.530439870559599\\
439	0.0368358742254542\\
440	-0.337035386563573\\
441	-0.0889357056040083\\
442	0.324262973952245\\
443	0.44273683281229\\
444	0.477637122259946\\
445	-0.119290865565604\\
446	-0.0104054744326032\\
447	0.0391271137975241\\
448	-0.0488755559221992\\
449	-0.029373235152224\\
450	-0.0722857018796355\\
451	0.273608481263978\\
452	0.0888643579571426\\
453	-0.218904502253755\\
454	-0.286777840149437\\
455	-0.107298976364117\\
456	-0.158237796001112\\
457	0.310713082531686\\
458	0.143627410045039\\
459	0.108575304460231\\
460	0.588054401085698\\
461	0.125624481436702\\
462	-0.223092113635411\\
463	0.081884763832074\\
464	0.155305111066639\\
465	0.244115187122445\\
466	-0.0603012834899418\\
467	0.307705068161288\\
468	0.00584326965317171\\
469	0.447789236689163\\
470	-0.107302543772094\\
471	0.14431797535652\\
472	0.0692081883163003\\
473	0.300441960444847\\
474	0.0519096458090519\\
475	0.0728125060311962\\
476	-0.173880119607556\\
477	0.338640383177941\\
478	0.507270758958238\\
479	0.101773774117899\\
480	-0.267299625769398\\
481	0.547888443998305\\
482	-0.162120215723629\\
483	0.268854238166861\\
484	-0.139111215469804\\
485	0.405099381005253\\
486	0.017369626202101\\
487	-0.140172604409092\\
488	-0.314407679758602\\
489	-0.42945515257275\\
490	-0.0790119208245361\\
491	-0.42388108719653\\
492	-0.22083579153586\\
493	-0.0720501222672856\\
494	0.263100978828832\\
495	0.101026586952142\\
496	0.0721951782229722\\
497	-0.364127477372184\\
498	-0.249602281223032\\
499	-0.554735260478247\\
500	-0.571299137630932\\
};
\end{axis}
\end{tikzpicture}%

%% file: fig_asy_sv.tikz
%
%
\definecolor{mycolor1}{rgb}{0.00000,0.00000,0.56250}%
\begin{tikzpicture}[scale=0.8]

\begin{axis}[%
width=6.248in,
height=3.566in,
at={(0.8in,0.481in)},
scale only axis,
separate axis lines,
every outer x axis line/.append style={black},
every x tick label/.append style={font=\color{black}},
xmin=-0.025,
xmax=0.015,
every outer y axis line/.append style={black},
every y tick label/.append style={font=\color{black}},
ymin=0,
ymax=90,
xlabel={Bins},
ylabel={Observations per bin},
axis background/.style={fill=white}
]
\addplot[ybar,bar width=04,draw=myblue,fill=myblue,area legend] plot table[row sep=crcr] {%
-0.0246268274622722	2.50773048756063\\
-0.0238276953282679	0\\
-0.0230285631942635	0\\
-0.0222294310602592	1.25386524378032\\
-0.0214302989262548	0\\
-0.0206311667922505	1.25386524378032\\
-0.0198320346582462	2.50773048756063\\
-0.0190329025242418	2.50773048756063\\
-0.0182337703902375	1.25386524378032\\
-0.0174346382562332	2.50773048756063\\
-0.0166355061222288	5.01546097512127\\
-0.0158363739882245	3.76159573134095\\
-0.0150372418542202	5.01546097512127\\
-0.0142381097202158	3.76159573134095\\
-0.0134389775862115	6.26932621890159\\
-0.0126398454522071	8.77705670646222\\
-0.0118407133182028	7.5231914626819\\
-0.0110415811841985	12.5386524378032\\
-0.0102424490501941	18.8079786567048\\
-0.0094433169161898	26.3311701193867\\
-0.00864418478218546	17.5541134129244\\
-0.00784505264818113	28.8389006069473\\
-0.00704592051417679	30.0927658507276\\
-0.00624678838017246	43.8852835323111\\
-0.00544765624616812	43.8852835323111\\
-0.00464852411216378	41.3775530447505\\
-0.00384939197815945	53.9162054825536\\
-0.00305025984415511	58.9316664576749\\
-0.00225112771015077	55.170070726334\\
-0.00145199557614644	61.4393969452355\\
-0.000652863442142101	72.7241841392584\\
0.000146268691862234	67.7087231641371\\
0.00094540082586657	86.5167018208419\\
0.00174453295987091	70.2164536516978\\
0.00254366509387524	65.2009926765765\\
0.00334279722787958	65.2009926765765\\
0.00414192936188391	50.1546097512127\\
0.00494106149588825	52.6623402387733\\
0.00574019362989259	45.1391487760914\\
0.00653932576389692	32.6004963382882\\
0.00733845789790126	27.585035363167\\
0.00813759003190559	20.0618439004851\\
0.00893672216590993	17.5541134129244\\
0.00973585429991427	13.7925176815835\\
0.0105349864339186	5.01546097512127\\
0.0113341185679229	7.5231914626819\\
0.0121332507019273	2.50773048756063\\
0.0129323828359316	0\\
0.0137315149699359	1.25386524378032\\
0.0145306471039403	5.01546097512127\\
};
\addplot [color=black,solid,forget plot]
  table[row sep=crcr]{%
-0.025	0\\
0.015	0\\
};
\end{axis}
\end{tikzpicture}%